 \documentclass[final,5p,times]{elsarticle}
\usepackage{amssymb,fullpage}
\usepackage{epstopdf}
\usepackage[english]{babel}
\usepackage{comment}
\biboptions{numbers,sort&compress}
\usepackage{lineno}
\usepackage{subfigure}
\usepackage{enumerate}
\usepackage{url}
\usepackage[colorlinks = true]{hyperref}
\usepackage{amsmath,amsthm,amssymb,amsbsy,bbm,bm}
\usepackage{mathdots}
\usepackage{paralist}
\usepackage{xcolor}
\usepackage{color}
\usepackage{graphicx}

\usepackage{fancyhdr}
\usepackage{cleveref}
\usepackage[english]{babel}
\usepackage{color}

\usepackage[framemethod=tikz]{mdframed}
\newtheorem{thm}{Theorem}%[section] %(If you want theorem numbered
\newtheorem{cor}{Corollary}%[section]
\newtheorem{defi}{Definition}%[section]
\newtheorem{lem}{Lemma}
\theoremstyle{remark}

\newcommand{\R}{\mathbb{R}}

\newcommand{\N}{\mathbb{N}}

\newcommand{\vct}[1]{\boldsymbol{#1}}
\newcommand{\mtx}[1]{\boldsymbol{#1}}
\newcommand{\T}{\mathrm{T}}

\newcommand{\sign}{\operatorname{sign}}
\newcommand{\set}[1]{\mathcal{#1}}
\newcommand{\domain}{\operatorname{dom}}

\newcommand{\calS}{\mathcal{S}}

\newcommand{\calP}{\mathcal{P}}
\newcommand{\calC}{\mathcal{C}}

\newcommand{\calG}{\mathcal{G}}

\newcommand{\ve}{\vct{e}}

\newcommand{\vq}{\vct{q}}

\newcommand{\vs}{\vct{s}}

\newcommand{\vu}{\vct{u}}
\newcommand{\vv}{\vct{v}}

\newcommand{\vx}{\vct{x}}
\newcommand{\vy}{\vct{y}}
\newcommand{\vz}{\vct{z}}
\newcommand{\vzero}{\vct{0}}

\newcommand{\mA}{\mtx{A}}
\newcommand{\mB}{\mtx{B}}

\newcommand{\mD}{\mtx{D}}
\newcommand{\mE}{\mtx{E}}

\newcommand{\mG}{\mtx{G}}

\newcommand{\mQ}{\mtx{Q}}

\newcommand{\mS}{\mtx{S}}

\newcommand{\mU}{\mtx{U}}
\newcommand{\mV}{\mtx{V}}
\newcommand{\mW}{\mtx{W}}
\newcommand{\mX}{\mtx{X}}

\newcommand{\mZ}{\mtx{Z}}

\newcommand{\mPhi}{\mtx{\Phi}}
\newcommand{\mPsi}{\mtx{\Psi}}

\newcommand{\mId}{{\bf I}}

\newcommand{\mzero}{{\bf 0}}

\newcommand{\setC}{\set{C}}

\newcommand{\setG}{\set{G}}

\newcommand{\setL}{\set{L}}

\newcommand{\setS}{\set{S}}

\graphicspath{{./figs/}}

\newlength{\imgwidth}
\usepackage{subfigure}
\usepackage{todonotes}
\usepackage{graphicx}
\usepackage{bm}
\usepackage{amsmath,amssymb,amsthm}  %~\\¿ÉÓÃÓÚÔÚ¶¨Àí¼°Ö¤Ã÷µÈ»µ¾³ÏÂµÄ»»ÐÐ
\usepackage{paralist}%[flushleft] Ê¹ÁÐ±í×ó¶ÔÆëµÄÃüÁî
\usepackage{algorithm}
\usepackage{algorithmic}        %ÓÃµ½µÄºê°ü£¬Òª×Ô¼º¸ÄÏÂ
\usepackage{multirow}
\renewcommand{\algorithmicrequire}{\textbf{Initialization:}}   %¸Ä³ÉºóÃæµÄÐ¡±êÌâ
\renewcommand{\algorithmicensure}{\textbf{Iteration:}}
\renewcommand{\algorithmicuntil}{\textbf{until convergence}}
\renewcommand{\algorithmiclastcon}{\textbf{Output:}} %Used for the algorithm list

\newcommand{\e}{\begin{equation}}
\newcommand{\ee}{\end{equation}}
\newcommand{\en}{\begin{equation*}}
\newcommand{\een}{\end{equation*}}
\newcommand{\eqn}{\begin{eqnarray}}
\newcommand{\eeqn}{\end{eqnarray}}
\newcommand{\bmat}{\begin{bmatrix}}
\newcommand{\emat}{\end{bmatrix}}
\newcommand{\BIT}{\begin{itemize}}
\newcommand{\EIT}{\end{itemize}}
\newcommand{\argmin}{\mathop{\rm argmin}}

\newcounter{oursection}

\usepackage{color}
\journal{Signal Processing}
\usepackage{color}
\usepackage{cleveref}
\begin{document}
	
	\begin{frontmatter}
		
		
		
\title{{Optimized Structured Sparse Sensing Matrices for Compressive Sensing}}
        % Optimized Sparse Sensing Matrices for Compressive Sensing		
		
		\author[TH]{Tao Hong}
		\ead{hongtao@cs.technion.ac.il}
		\author[XL]{Xiao Li}
		\ead{xli@ee.cuhk.edu.hk}
		\author[ZZ]{Zhihui Zhu}%\corref{cor1}
		\ead{zzhu29@jhu.edu}
		\author[QL]{Qiuwei Li}
		\ead{qiuli@mines.edu}
		
		%\cortext[cor1]{Corresponding author.}
		%\author{Michael~Shell,~\IEEEmembership{Member,~IEEE,}
		%        John~Doe,~\IEEEmembership{Fellow,~OSA,}
		%        and~Jane~Doe,~\IEEEmembership{Life~Fellow,~IEEE}% <-this % stops a space
		% % revised December 27, 2012.
		\address[TH]{Department
			of Computer Science, Technion - Israel Institute of Technology, Haifa, 32000, Israel.}%mzib@cs.technion.ac.il
		%\thanks{Y. Wang is with the Department of Electrical Engineering, Shanghai Jiao Tong University, 200240, China (e-mail: Queen\_Wang@sjtu.edu.cn).}
	
	\address[XL]{Department
		of Electronic Engineering, The Chinese University of Hong Kong, Shatin, NT, Hong Kong.}
		
	\address[ZZ]{Center of Imaging Science, Johns Hopkins University, Baltimore, MD 21218 USA}
		
    \address[QL]{Department of Electrical Engineering,
			Colorado School of Mines, Golden, CO 80401 USA.} %GA, 30332 USA e-mail: (see http://www.michaelshell.org/contact.html).}% <-this % stops a space
		%  \address{}
		
		%\address{}
		
		\begin{abstract}
		{ We consider designing a robust structured sparse sensing matrix consisting of a sparse matrix with a few non-zero entries per row and a dense base matrix for capturing signals efficiently.} We design the robust structured sparse sensing matrix through minimizing the distance between the Gram matrix of the equivalent dictionary and the target Gram of matrix holding small mutual coherence. Moreover, a regularization is added to enforce the robustness of the optimized structured sparse sensing matrix to the sparse representation error (SRE) of signals of interests. An alternating minimization algorithm with global sequence convergence is proposed for solving the corresponding optimization problem. Numerical experiments on synthetic data and natural images show that the obtained structured sensing matrix results in a higher signal reconstruction than a random dense sensing matrix.
		
		\end{abstract}
		
		\begin{keyword}
		Compressive sensing  \sep structured sensing matrix\sep sparse sensing matrix \sep  mutual coherence \sep sequence convergence.
			
			
		\end{keyword}

	\end{frontmatter}
	
\section{Introduction}\label{sec:introduction}
Compressive sensing (CS) supplies a paradigm of  joint compression
and sensing signals of interest \cite{CandesRombergTao2006RobustUncertaintyPrinciples,CandesWakin2008IntroductionCS}. A CS system contains two main ingredients: a {\em sensing matrix} $\bar\mPhi\in\R^{M\times N}$ ($M\ll N$) which compresses a signal $\vx$ via
$\vy = \bar\mPhi \vx$
and a  {\em dictionary} $\bar\mPsi\in\R^{N\times L}$ ($L\geq N$) that captures the sparse structure of the signal. In particular, we say $\vx\in\R^N$ is sparse if it  can be represented with a few columns of $\bar\mPsi$:
\e
\vx = \bar\mPsi \vs + \ve = \sum_{\ell}\bar\mPsi(:,\ell) \vs(\ell) + \ve,
\label{eq:sparse representation x}
\ee
where $\|\vs\|_0\leq K$ with $K\ll N\leq L$.\footnote{Throughout this paper, MATLAB notations are adopted: $\mQ(m,:), \mQ(:,k)$ and $\mQ(i,j)$ denote the $m$th row, $k$th column, and $(i,j)$th entry of the matrix $Q$; $\vq(n)$ denotes the $n$th entry of the vector $\vq$. $\|\cdot\|_0$ is used to count the number of nonzero elements.} The term $\ve$ is referred to as the sparse representation error (SRE) of $\vx$ under $\bar\mPsi$. If $\ve$ is nil, we say $\vx$ is exactly sparse.

% CS is a mathematical framework that efficiently compresses a signal vector  $\vx\in\R^{N}$  by a linear measurement vector $\vy \in \R^{M}$ of the form
% \e
% \vy = \bar\mPhi \vx,
% \label{eq:y}\ee
% where $M\ll N$ and  is a carefully chosen.

% As $M\ll N$, additional constraints or structures on the signal vector $\vx$ have to be posed in order to recover it from the measurement $\vy$. {\em Sparsity} is a widely used such structure of the signal vector in CS, where we assume the original signal $\vx$ can be expressed as a linear combination of few atoms from a well-chosen dictionary:
% \begin{align}
% \vx = \sum_{\ell=1}^L \vpsi_\ell s_\ell = \bar\mPsi \vs.
% \label{eq:represet x}\end{align}
% Here $\bar\mPsi :=\begin{bmatrix}\vpsi_1 & \vpsi_2 & \cdots \vpsi_L   \end{bmatrix}\in\R^{N\times L} $ is called the {\em dictionary} of the CS system and the entries of $\vs\in\R^{L}$ are referred to as coefficients.
% %When the coefficients have a small fraction of nonzero values or decay quickly, one can form a highly accurate and concise representation of the original signal using just a small number of atoms.
% We say $\vx$ given by \eqref{eq:represet x} $K$-sparse in $\bar\mPsi$ if $\|\vs\|_0 = K$, where $\|\cdot\|_0$ denotes the number of non-zero elements.

The choice of dictionary $\bar\mPsi$ depends on the signal model and traditionally it is chosen to concisely capture the structure of the signals of interest, e.g., the Fourier matrix for frequency-sparse signals, and a multiband modulated Discrete Prolate Spheroidal Sequences (DPSS's) dictionary for sampled multiband signals \cite{ZhuWakin2015MDPSS}. Furthermore, we can also learn a dictionary from a set of representative signals (training data) called dictionary learning \cite{engan1999MOD,aharon2006ksvd,li2017new}.
%with algorithms including the method of optimal directions (MOD)~\cite{engan1999MOD}, K-singular value decomposition (K-SVD)~\cite{aharon2006ksvd} based algorithms and the method for designing incoherent sparsifying dictionary~\cite{li2017new}.

In CS, the sensing matrix $\bar\mPhi$ is used to preserve the useful information contained in the signal $\vx$ such that it is possible to recover $\vx$ from its low dimensional measurements $\vy = \bar\mPhi\vx$. It has been shown that if the \emph{equivalent dictionary} $\bar\mPhi\bar\mPsi$ satisfies the restricted isometry property (RIP), the sparse vector $\vs$ in \eqref{eq:sparse representation x} can be exactly recovered from $\vy$~\cite{baraniuk2008simple,CandesRombergTao2006RobustUncertaintyPrinciples}.
%with methods based on convex optimization~\cite{CandesRombergTao2006RobustUncertaintyPrinciples,chen1998atomicDecomposition,donoho2003optimallySparseRepl1} or greedy algorithms~\cite{blumensath2008iterativeThresholding,mallat1993matchingPursuit,needell2009cosamp,tropp2004greed}.
Although random matrices satisfy the RIP with high probability \cite{baraniuk2008simple}, confirming whether a general matrix satisfies the RIP is NP-hard \cite{bandeira2013certifying}. Alternatively, mutual coherence, another measure of sensing matrices that is much easier to verify, has been introduced in practice to quantify and design sensing matrices~\cite{elad2007optimized,duarte2009learning,abolghasemi2012gradient,li2013projection,chen2013projection,li2015designing,bai2015alternating,hong2016efficient,hong2017SP,li2017gradient,zhu2018collaborative}. %{\color{blue} However, the computation complexity of implementing sensing step is not considered in these methods.}%This approach has also been extended for block-sparse CS~\cite{zelnik2011sensing,lishuang2013projection}.

{ Structured sensing matrices (e.g., Toeplitz matrices and sparse matrices) have been proposed \cite{yin2010practical,zhang2010compressed,dias2013comparative,sun2013sparse,fan2014toeplitz,gan2008fast,rauhut2010compressive} to reduce the computational complexity of sensing signals in hardware (such as digital signal processor and FPGA) \cite{chen2012design,dou200564}, or  applications like electrocardiography (ECG) compression~\cite{mamaghanian2011compressed} and data stream computing~\cite{gilbert2010sparse}. A Toeplitz matrix can be implemented efficiently to a vector by the fast Fourier transform (FFT). The advantage of sparse sensing matrix over a regular one is that it contains fewer non-zero elements per row and thus can significantly reduce the number of multiplication units  for practical applications.} However, similar to a random sensing matrix, a random sparse one is less competitive than an optimized sensing matrix regarding signal recovery accuracy. 

%Based on our knowledge, no effort has been devoted to optimize a sparse sensing matrix as has been done for a dense one by minimzing the mutual coherence in \cite{elad2007optimized,duarte2009learning,abolghasemi2012gradient,li2013projection,chen2013projection,li2015designing,bai2015alternating,hong2016efficient,hong2017SP,li2017gradient,zhu2018collaborative}.

\begin{figure*}[!htb]
	\centering
	\includegraphics[width = 12cm]{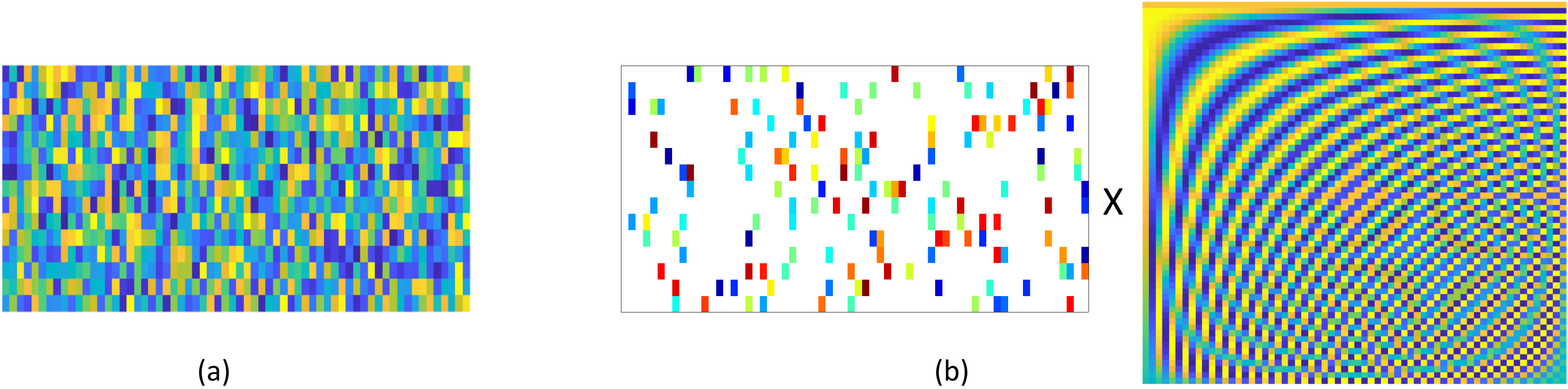}
	\caption{(a) a random Gaussian matrix; (b) a structured sparse sensing matrix consists of a sparse sensing matrix and a { base sensing} matrix}\label{fig:sprseMatrix}
\end{figure*}

Motivated by this, we consider the design of a structured sparse sensing matrix that can not only efficiently compress signals but also has similar performance as the dense ones. Specifically, we attempt to design a structured sparse sensing matrix via enhancing the mutual coherence (defined in \eqref{eq:define:mutualcoherence}) property of the equivalent dictionary, $\bar{\mPhi}\bar{\mPsi}$.  Our main contributions are stated as follows:
\vspace{-0.1cm}
\begin{itemize}
	\item We propose a framework for designing a \emph{structured sparse sensing matrix} by decreasing the mutual coherence of the equivalent dictionary. As shown in \Cref{fig:sprseMatrix}, the structured  sparse sensing matrix consists of $\widetilde\mPhi\mA$ where $\widetilde \mPhi\in\R^{M\times N}$ is a {\em row-wise sparse matrix} while $\mA\in\R^{N\times N}$ is referred to a {\em base sensing matrix} that can be implemented with linear complexity to a signal. { In general, the choice of $\mA$ depends on the practical situations, e.g., we choose $\mA$ as a DCT matrix when used for natural images with a dictionary learned by the KSVD algorithm \cite{aharon2006ksvd}. For some cases, one may simply set $\mA$ as an identity matrix, giving a sparse sensing matrix. To our knowledge, this work is the first attempt to optimize a (structured) sparse sensing matrix by minimizing the mutual coherence.}
	\vspace{-0.2cm}    
	% Moreover the resultant sparse sensing matrix is also robust to SRE, which is generally not negligible in practical applications like image processing~\cite{li2015designing}. 
	\item  We provide an alternating minimization algorithm for solving the formulated nonconvex nonsmooth optimization problem (see \eqref{eq:sparse sensing problem}). Despite the nonconvexity and nonsmoothness, we perform a rigorous convergence analysis to show that the sequence of iterates generated by our proposed algorithm with random initialization converges to a critical point. %\eqref{eq:unify}.
	\vspace{-0.2cm}
	\item {Experiments on natural images show that the obtained structured sensing matrix---with or without $\mA$---outperforms a random dense sensing matrix. It is of interest to note that by setting $\mA$ as the DCT matrix, the optimized structured sensing matrix has almost identical performance in terms of Peak Signal to Noise Ratio (PSNR) as the optimized dense sensing matrix, see \Cref{fig:PSNR high dim}}.
	%         An alternating minimization-based algorithm is proposed for solving the formulated nonconvex nonsmooth optimization problem to the optimal design problem which has been popularly utilized in designing sensing matrix~\cite{elad2007optimized,duarte2009learning,abolghasemi2012gradient,li2013projection,li2015designing,hong2016efficient}. 
	%However, the convergence of these algorithms is usually not guaranteed. Following this concern, we provide the convergence analysis of the proposed algorithm. 
\end{itemize}
\vspace{-0.2cm}

The outline of this paper is given as follows. We review the previous approaches in robust sensing matrix design in Section~\ref{sec:preliminary}. In Section \ref{sec:design problem}, a framework for designing a structured sparse sensing matrix is proposed with the mutual coherence behavior of the equivalent dictionary and the SREs of the signals being considered simultaneously. An alternating minimization algorithm for solving the optimal design problem with a rigorous convergence analysis is provided in Section~\ref{sec:proposed algorithm}. We validate the performance of the obtained structured sensing matrix on both synthetic data and real images in Section~\ref{sec:simulations}. Conclusions are given in Section \ref{sec:Conclusions}.

\section{Preliminaries}
\label{sec:preliminary}
In this section, we will brief the definition of mutual coherence to CS and introduce the previous work on designing robust sensing matrices.
%\vspace{-0.3cm}
\subsection{Mutual Coherence}
The \emph{mutual coherence} of $\mQ\in\R^{M\times L}$ is defined as
\begin{equation}
\mu(\mQ)\triangleq \max_{1\leq i \neq j \leq L}\frac{|\vq_i^{\cal
T}\vq_j|}{\|\vq_i\|_2 \|\vq_j\|_2}\geq\underline\mu \triangleq \sqrt{\frac{L-M}{M(L-1)}},
\label{eq:define:mutualcoherence}
\end{equation}
where $\vq_i$ is the $i$th column of $\mQ$ and $\underline \mu$ is the lower bound of $\mu(\mQ)$ called Welch Bound \cite{strohmer2003grassmannian}.  {The connection between the mutual coherence and the RIP is given in \cite[Section 5.2.3]{Elad:2010book}. Roughly speaking, the smaller mutual coherence, the better the RIP.}

With the measurements $\vy = \bar\mPhi \vx$ and the prior information that $\vx$ is sparse in $\bar\mPsi$, we can recover the signal as $\widehat \vx = \bar\mPsi \widehat \vs$ where\footnote{Here $\|\cdot\|_2$ denotes the $l_2$ norm of a vector.}
\begin{equation}
\widehat \vs = \argmin_{\vs} ~ \|\vy - \bar\mPhi\bar\mPsi  \vs\|_2^2~~~~\text{s.t.}~~~  \|\vs\|_0 \leq K 
\label{eq:sparserecov}
\end{equation}
which can be exactly or approximately solved via convex methods~\cite{CandesRombergTao2006RobustUncertaintyPrinciples,chen1998atomicDecomposition,donoho2003optimallySparseRepl1} or greedy algorithm~\cite{tropp2004greed}, e.g., the orthogonal marching pursuit (OMP). It is shown in~\cite{tropp2004greed} that OMP can stably find $\vs$  (and hence obtain an accurate estimation of $\vx$) if
\begin{equation}
 K <
\frac{1}{2}\left[1+\frac{1}{\mu(\bar\mPhi\bar\mPsi)}\right].\label{mu-K}
\end{equation}

%However, the mutual coherence is not easy to be directly minimized and also it is too preserved as it measures the worst-case of CS systems. Instead of directly minimizing the mutual coherence, Elad \cite{elad2007optimized}  considered the $t$-averaged mutual cohere which focuses on the average behavior: 
%\e
%\mu_t(\bar {\mQ}) = \frac{\sum_{\forall i\neq j}\left(|\bar {\bm q}_i^\mathcal T\bar{\bm q}_j|\geq t\right)\cdot|\bar {\bm q}_i^\mathcal T\bar{\bm q}_j|}{\sum_{\forall i\neq j}\left(|\bar {\bm q}_i^\mathcal T\bar{\bm q}_j|\geq t\right)},
%\label{eq:tav:MC}
%\ee
%where $\bar{\mQ}$ is a matrix whose columns are normalized.
\subsection{Optimized Robust Sensing Matrix \cite{li2015designing,hong2016efficient}}
Motivated by \eqref{mu-K}, abundant efforts have been devoted to design the sensing matrix via minimizing the mutual coherence $\mu(\bar\mPhi\bar\mPsi)$, including a subgradient projection method \cite{lu2014design}, and the ones based on alternating minimization. \cite{elad2007optimized,abolghasemi2012gradient,li2013projection}. Experiments on synthetic data indicate that the obtained  sensing matrices give much better performance than the random one when the signals are exactly sparse, i.e., $\ve = \vzero$ in \eqref{eq:sparse representation x}.

However, it was recently realized that an optimized sensing matrix obtained by minimizing the mutual coherence is not robust to SRE in \eqref{eq:sparse representation x} and thus the corresponding CS system yields poor performance \cite{li2015designing}. In particular, the SRE  always exists in the practical signals of interests,  even representing them via a learned dictionary \cite{aharon2006ksvd}. Let $\mX\in \R^{N\times J}$ be a set of training data and $\mS$ consist of the sparse coefficients of $\mX$ in $\bar\mPsi$:
$\mX = \bar\mPsi \mS + \mE$
where $\|\mS(:,j)\|_0\leq K, \forall j$. Then, in \cite{li2015designing,hong2016efficient}, the SRE matrix 
\e
\mE := \mX - \bar\mPsi \mS
\label{eq:SRE E}\ee
is utilized as the regularization to yield a robust sensing matrix.

Denote by $\calG_\xi$ the set of relaxed equiangular tight frame (ETF) Gram matrices:
\begin{align}
\calG_{\xi} = \bigg\{\mG\in\mathbb S^{L\times L}: \mG(i,i) = 1,\forall i, \max_{i\neq j}|\mG(i,j)|\leq \xi \bigg\},
\label{eq:define G}\end{align}
where $\xi\in[0,1)$ is a pre-set threshold and usually chosen as $0$ or $\underline\mu$ \cite{abolghasemi2012gradient,li2013projection,li2015designing,hong2016efficient} and $\mathbb S^{L\times L}$ denotes a set of real $L\times L$ symmetric matrices. Then the sensing matrices proposed in \cite{li2015designing,hong2016efficient} are optimized by solving the following optimization problem\footnote{$\|\cdot\|_F$ represents the Frobenius norm.}:
\begin{equation}\underset{\bar\mPhi, \mG\in\calG_\xi } {\text{min}}~ ||\mG - \bar\mPsi^\T\bar\mPhi^\T\bar\mPhi\bar\mPsi||^2_F + \lambda \|\bar\mPhi\mE\|_F^2, \label{eq:existing approaches}\end{equation}
where the first term is utilized to control the average mutual coherence of the equivalent dictionary, the second term $\|\bar\mPhi\mE\|_F^2$ is  a regularization to make the sensing matrix robust to SRE, and $\lambda\geq 0$ is the trade-off parameter to balance these two terms. Compared with previous work, simulations have shown that the obtained sensing matrices by \eqref{eq:existing approaches} achieve the highest signal recovery accuracy when the SRE exists \cite{li2015designing}.

%We finally note that $\mG_t = \bar\mPhi^{\cal T}\bar\mPhi$ is utilized in \cite{duarte2009learning} where the designed sensing matrix is also robust to SRE.

\section{Optimized Structured Sparse Sensing Matrix}
\label{sec:design problem}
In this section, we consider designing a structured sensing matrix by taking into account the complexity of signal sensing procedure, robustness against the SRE  and the mutual coherence of the equivalent dictionary simultaneously.

As mentioned above, in applications like ECG compression~\cite{mamaghanian2011compressed}, data stream computing~\cite{gilbert2010sparse} and hardware implementation~\cite{chen2012design}, the classical CS system with a dense sensing matrix $\bar\mPhi$ encounters computational issues. Indeed, merely applying a sensing matrix $\bar\mPhi\in\R^{M\times N}$ to capture a length-$N$ signal has the computational complexity of $\mathcal O(MN)$. Moreover, in applications like image processing, one often partitions the image into a set of patches of small size (say $8\times 8$ patches, hence $N = 64$) to make the problem computationally tractable. However, the recent work in dictionary learning~\cite{sulam2016trainlets} and sensing matrix design~\cite{hong2017SP} has revealed that larger-size patches (say $64\times 64$ patches, hence $N = 4096$) lead to better performance for image processing like image denoising and compression. All of these enforce us to reduce the complexity of sensing a signal. 

An approach to tackle this computational difficulties is to impose  certain structures into the sensing matrix $\bar\mPhi$. One of such structures is the sensing matrix consisting of a sparse matrix and a base matrix that both can be efficiently implemented to sensing signals:
\e
\bar\mPhi = \widetilde\mPhi \mA,
\label{eq:structured sparse matrix}
\ee
where $\mA\in\R^{N\times N}$ is referred to as a {\em base sensing matrix} and $\widetilde \mPhi\in\R^{M\times N}$ is a row-wise sparse matrix. 

{To maintain the original purpose for reducing sensing complexity of $\vy = \bar\Phi \vx$, we restrict the choices of the base sensing matrix  $\mA$ to be either \emph{identity matrix} or the one that \emph{admits fast multiplications} like DCT matrix. We also note that the choice of base  sensing matrix $\mA$ should depend on specific applications, e.g., we can set $\mA$ to be a DCT matrix in image processing task \cite{rubinstein2010double}}.
Rewrite \eqref{eq:structured sparse matrix} as
\e
\bar\mPhi^\T  = \mA^\T \widetilde \mPhi^\T
\label{eq:structured sparse matrix 2}
\ee 
and view $\widetilde \mPhi^\T$ as the sparse representation of $\bar\mPhi^\T$ in $\mA^\T$. Thus, similar to \eqref{eq:sparse representation x} where we call $\vx$ is sparse (in $\bar\mPsi$) though itself is not sparse, we also say that $\bar\mPhi$ in \eqref{eq:structured sparse matrix} is a sparse sensing matrix (in $\mA$).\footnote{Without any confusion, we call our sensing matrix $\bar{\mPhi}$ as the structured or sparse sensing matrix in the rest of the paper.} Note that the structure of \eqref{eq:structured sparse matrix 2} also appears in the double-sparsity dictionary learning task~\cite{rubinstein2010double}, the dictionary $\bar\mPsi = \mA \widetilde \mPsi$ with $\widetilde\mPsi$ being overcomplete but column-wise sparse.

Note that the approach shown in \cite{li2015designing,hong2016efficient} requires the explicit formulation of the SRE matrix $\mE$ (defined in \eqref{eq:SRE E}) which can be huge or need extra effort to obtain in some applications, like image processing with a wavelet dictionary that no typical training data is available \cite{hong2017SP}. Let us draw each  column of $\mE$ from an  independently and identically distributed (i.i.d.) Gaussian distribution of mean zero and covariance $\sigma^2\mId$. Then $\frac{\|\bar\mPhi \mE\|_F^2}{J}$ converges in probability and almost surely to $\sigma^2\|\bar\mPhi\|_F^2$ when the number of training samples $J$ approaches to $\infty$ \cite{hong2017SP}. Thus we can get rid of the SRE matrix $\mE$ by minimizing $\|\bar\mPhi\|_F$ directly to yield a sensing matrix that is robust to the SRE.

Now our goal is to find a structured sensing matrix $\bar\mPhi = \widetilde\mPhi\mA$ such that it is robust to the SRE of the signals and the  equivalent dictionary $\bar\mPhi \bar\mPsi$ has a small mutual coherence.  So the corresponding sparse matrix is obtained via solving
\begin{align}
%\begin{split}
\{\widetilde\mPhi,\widetilde \mG\} = & \argmin_{\mPhi,\mG\in\calG_\xi} ||\mG - \bar\mPsi^\T\mA^\T\mPhi^\T\mPhi\mA\bar\mPsi||^2_F+\lambda \|\mPhi \mA\|_F^2\nonumber\\
&~~~~~\text{s.t.}~~~ \|\mPhi(m,:)\|_0 \leq \kappa, \ \forall \ m,
%\end{split}
 \label{eq:unify}
 \end{align}
where $\kappa$ denotes the number of non-zero elements in each row of the sensing matrix.

%3\footnote{When $\mA$ does not equal to identity, we first implement $\hat\vy = \mA\vx$ which costs $O(N\log N)$ and then $\vy = \widetilde \mPhi \hat\vy$ that is of $O(M\kappa)$ complexity.}.} 

 {Without the base matrix $\mA$, the complexity of $\bar\mPhi \vx = \widetilde \mPhi  \vx$ is $O(M\kappa)$ which is the same as the one shown in \cite{sun2013sparse}.  Thus, there is a tradeoff in adding the base matrix $\mA$ since the base matrix may improve the performance, but also increase the computational complexity. Fortunately, by choosing as a DCT matrix, it only slightly increase the computation of $O(N\log N)$ which is small or comparable to $O(M\kappa)$.} { Moreover, in some cases, $\mA\vx$ can be implemented with complexity $O(N)$. For example, if $\mA$ is an orthogonal matrix and we decompose it into a series of Givens rotation matrices \cite{dita2003factorization}. This is a significant reduction of computational complexity compared with a dense sensing matrix requiring a complexity of $O(MN)$ to sense a signal when $N$ is large  and $N\gg \kappa$. We show the difference between $O(MN)$ and $O(N\log N + M\kappa)$ in \Cref{fig:comparcomplexity} with various $\kappa$ and $N$.}

\begin{figure}[!htb]
	\centering
	\includegraphics[width=0.5\textwidth]{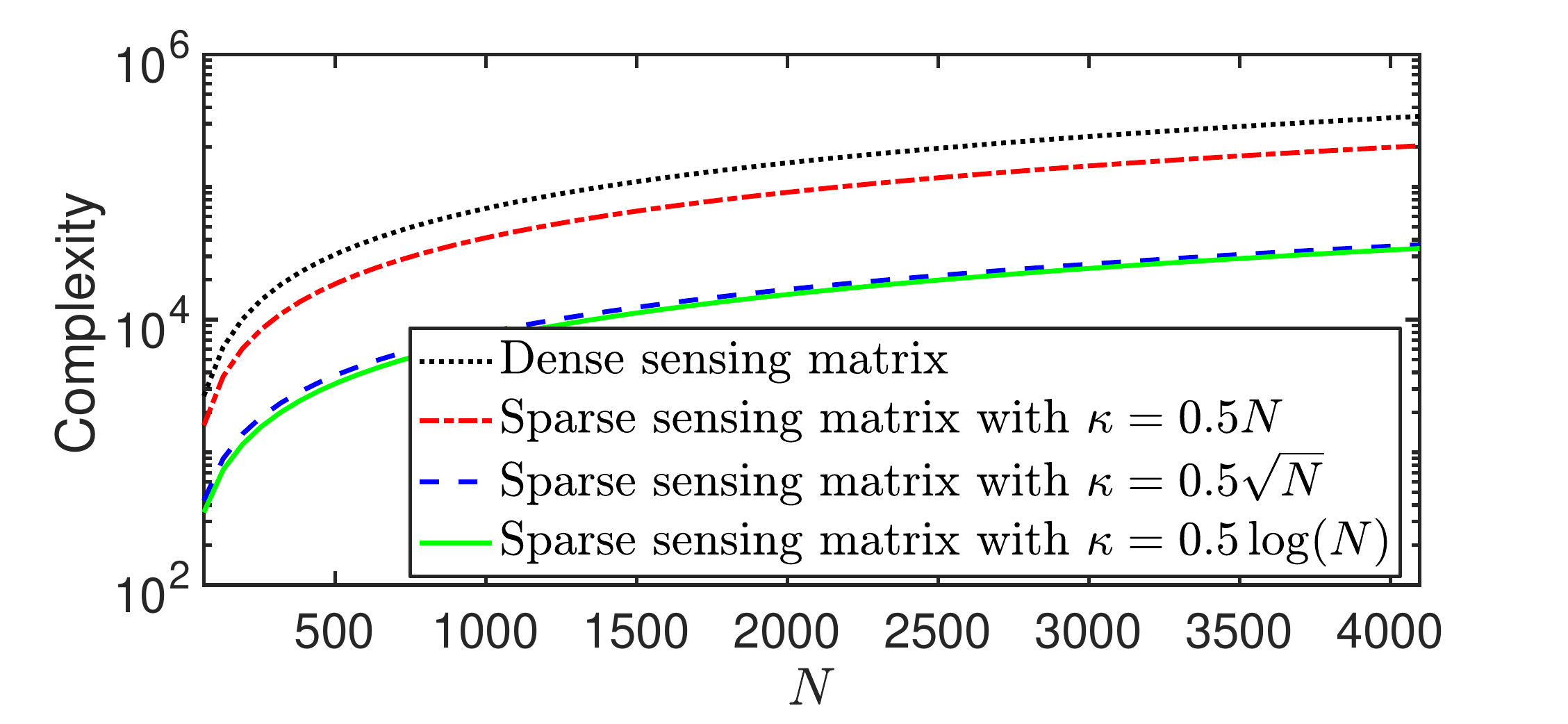}
	\caption{Illustration of the complexity  $O(MN)$ (the dotted black line) and $O(N\log N + M\kappa)$ (the other three lines indicating different $\kappa$) with $M = 10\log N$.}\label{fig:comparcomplexity}
\end{figure}

\section{Proposed Algorithm for Designing Structured Sensing Matrix}
\label{sec:proposed algorithm}
Aside from the facts that $\bar\mPhi$ is parameterized by $\widetilde \mPhi \mA$ and $\|\bar\mPhi\mE\|_F^2$ is replaced by $\|\bar\mPhi\|_F^2$, \eqref{eq:unify} differs from \eqref{eq:existing approaches} in that the former has a sparse constraint on the rows of $\widetilde \mPhi$. Note that such a constraint makes \eqref{eq:unify} highly nonconvex. In this section, we suggest utilizing alternating projected gradient method to address \eqref{eq:unify}. Moreover, we also provide a rigorous convergence analysis of the proposed algorithm.

\subsection{Proposed Algorithm for Designing Structured Sensing Matrix}
Assume $\mA$ is not null and rewrite~\eqref{eq:unify} as
\begin{equation}
\begin{split}
&\min_{\mPhi,\mG} f(\mPhi,\mG)= ||\mG - \mPsi^\T\mPhi^\T\mPhi\mPsi||^2_F+\lambda \|\mPhi\|_F^2\\
&~~~\text{s.t.}~~ \|\mPhi(m,:)\|_0 \leq \kappa, \ \forall \ m, \ \mG\in\calG_{\xi}
\label{eq:sparse sensing problem}
\end{split}
\end{equation}
where $\mPsi = \mA\bar\mPsi$.
Let $\calP_{\calG_{\xi}}:\R^{L \times L}\rightarrow \R^{L \times L}$ denote an orthogonal projector onto the set $\calG_{\xi}$:
\[
\left(\calP_{\setG_\xi}(\mG)\right)(i,j)=\left\{\begin{matrix}
1, & i = j,\\
\sign(\mG(i,j))\min(|\mG(i,j)|,\xi), & i\neq j,
\end{matrix}\right.
\] where $\sign(\cdot)$ denotes the sign function. 
Firstly the solution of minimizing $f$ in terms of $\mG$ with fixed $\mPhi$ is given by
\e
\widehat \mG = \argmin_{\mG\in\calG_{\xi}}f(\mPhi,\mG) = \calP_{\calG_{\xi}} (\mPsi^\T\mPhi^\T\mPhi\mPsi).
\label{eq:problem G}
\ee

Now we consider solving \eqref{eq:sparse sensing problem} in terms of $\mPhi$ with fixed $\mG$:
\e
\min_{\mPhi} f(\mPhi,\mG)~~~ \text{s.t.}\ \|\mPhi(m,:)\|_0 \leq \kappa, \ \forall \ m.
\label{eq:problem Phi}
\ee %\footnote{This is also numerically verified in \cite{hong2017SP}.} 
Without the sparsity constraint, the recent works~\cite{Zhu2017,li2016nonconvex} have shown that any gradient-based algorithms can provably solve $\min_{\mPhi}f(\mPhi,\mG)$. Thus, we suggest utilizing the projected gradient descent (PGD) to solve \eqref{eq:problem Phi} with the sparsity constraint. The gradient of $f(\mPhi,\mG)$ with respect to $\mPhi$ is:
\e
\nabla_{\mPhi} f(\mPhi,\mG)=2\lambda\mPhi-4\mPhi\mPsi\mG\mPsi^\T+4\mPhi\mPsi\mPsi^\T\mPhi^\T\mPhi\mPsi\mPsi^\T.
\label{eq:d Phi}
\ee
For convenience, let $\calS_\kappa$ denote the set of matrices which have at most $\kappa$ non-zero elements in each row:
\[
\calS_{\kappa}\triangleq\left\{\mZ\in\R^{M\times N}:\|\mZ(m,:)\|_0\leq \kappa, \ \forall\ m \right\}.
\]
Denote
$\calP_{\cal S_{\kappa}}:\R^{M\times N}\rightarrow\R^{M\times N}$ as an orthogonal projector on the set of $\calS_{\kappa}$:  for any $M\times N$ input matrix, that keeps the largest $\kappa$ absolute values of each row. %\footnote{If there exist more than one choice to the $\kappa$ largest components, we pick any one of them.} 
So, in the $k$th step, we update $\mPhi$ as
\e
\mPhi_k \in \calP_{\calS_\kappa}\left(\mPhi_{k-1} - \eta\nabla f(\mPhi_{k-1},\mG_{k-1})\right).
\label{eq:update Phi}
\ee
we choose an arbitrary one if there exist more than one projections. 
%\footnote{We can either pick the initialization $\mPhi_0$ as a random matrix, or the one given in~\cite[Theorem 2]{li2013projection}.}

We summarize the proposed alternating minimization for solving \eqref{eq:sparse sensing problem} in Algorithm \ref{alg: for G and Phi}.  Note that alternating minimization-based algorithm has been popularly utilized for designing sensing matrix~\cite{elad2007optimized,duarte2009learning,abolghasemi2012gradient,li2013projection,li2015designing,hong2016efficient}. However, the convergence of these algorithms is usually neither ensured nor seriously considered. In the following, we provide the rigorous convergence analysis of the proposed Algorithm \ref{alg: for G and Phi}.

\begin{algorithm}[!htb]
	\caption{Algorithm for Designing Sparse Sensing Matrix}% with considering projection noise}
	\label{alg: for G and Phi}
	\begin{algorithmic}[1]
		\REQUIRE ~\\
		Initial value $\mPhi_0$, the number of maximal iterations $Iter_{max}$, step size $\eta$, the sparsity level $\kappa$ and the given trade-off parameter $\lambda$.
		\lastcon ~\\          %OUTPUT
		Sparse sensing matrix $\mPhi_{Iter_{max}}$.
		%\ENSURE
		\STATE $k\leftarrow 1$
		\WHILE {$k\leq Iter_{max}$}
		\STATE Update $\mPhi$: $\mPhi_k \in \calP_{\calS_\kappa}(\mPhi_{k-1} - \eta\nabla_{\mPhi} f(\mPhi_{k-1},\mG_{k-1}))$
		\STATE Update $\mG$: $\mG_k = \calP_{\calG_{\xi}} (\mPsi^\T\mPhi_{k}^\T\mPhi_{k}\mPsi)$
		\STATE $k\leftarrow k+1$
		%\STATE {}
		\ENDWHILE
	\end{algorithmic}
\end{algorithm}
%{\noindent \bf Remark 3.1:}
%\begin{itemize}
%\item When $\calG_{t}$ consists of a single element (like $\calG_{t} = \{\mId\}$), the problem \eqref{eq:unify 3} reduces to \eqref{eq:problem Phi} and hence the spare sensing matrix can be simply obtained by {\bf Algorithm \ref{alg:PGD}}.
%\item We note that \eqref{eq:problem Phi} is a sparisty-constrianed problem and thus we can solve it by many other algorihtms develped for CS, like OMP and convex methods (by replacing the sparsity  with $\ell_1$ norm).
%\end{itemize}

\subsection{Convergence Analysis}
Transfer \eqref{eq:sparse sensing problem} into the following unconstrained problem
\begin{equation}
\begin{split}
\min_{\mPhi,\mG} \rho(\mPhi,\mG):= f(\mPhi,\mG)+ \delta_{\setS_\kappa}(\mPhi) + \delta_{\setG_\xi}(\mG),
\label{eq:unconstrainded problem}
\end{split}
\end{equation}
where $\delta_{\setS_\kappa}(\mPhi) = \begin{cases}0, & \mPhi \in \setS_\kappa,\\ \infty, & \mPhi\notin \setS_\kappa \end{cases}$ is the indicator function (and similarly for $\delta_{\setG_\xi}(\mG)$). Clearly, \eqref{eq:unconstrainded problem} is equivalent to the original constrained problem \eqref{eq:sparse sensing problem}. {Compared with \eqref{eq:sparse sensing problem}, it is easier to take the subdifferential for \eqref{eq:unconstrainded problem} since it has no constraints. Thus, in the sequel, we focus on \eqref{eq:unconstrainded problem} since the convergence analysis mainly involves the subdifferential.}

Note that by updating $\mG$ with \eqref{eq:problem G}, $\rho(\mPhi_k,\mG_k)\leq \rho(\mPhi_k,\mG_{k-1})$.\footnote{This inequality is shown in \ref{sec:prf thm subsequence convergence}.} Following, we show the objective function is decreasing by updating the sensing matrix $\mPhi$. Denote $\rho_0 = \rho(\mPhi_0,\mG_0)$ and consider the sublevel set of $\rho$:
\[
\setL_{\rho_0} = \left\{(\mPhi,\mG): \rho(\mPhi,\mG)\leq \rho_0, \mG\in\calG_{\xi}, \mPhi\in\setS_\kappa\right\}.
\]
It is clear that for any point $(\mPhi,\mG)\in \setL_{\rho_0}$, $\|\mG\|_F$ is finite since $\mG\in\setG_\xi$ and
$\|\mPhi\|_F$ is finite since $\rho\rightarrow \infty$ when $\|\mPhi\|_F\rightarrow \infty$. Then with simple calculation, we have that both $\nabla_{\mPhi} f(\mPhi,\mG)$ and $\nabla_{\mG} f(\mPhi,\mG)$ are Lipshitz continuous, % for all % $(\mPhi,\mG)\in \setL_{\rho_0}$,
\e\begin{split}
	&\|\nabla_{\mPhi} f(\mPhi,\mG) - \nabla_{\mPhi} f(\mPhi',\mG)\|_F \leq L_c \|\mPhi - \mPhi'\|_F\\
	& \|\nabla_{\mG} f(\mPhi,\mG) - \nabla_{\mG} f(\mPhi,\mG')\|_F \leq L_c \|\mG - \mG'\|_F
\end{split}\label{eq:Lipschitz gradient}\ee
for all $(\mPhi,\mG),(\mPhi',\mG),(\mPhi,\mG')\in\setL_{\rho_0}$. Here $L_c>0$ is the corresponding Lipschitz constant. A direct consequence of the Lipschitz continuous is as follows.
\begin{lem} For any $L\geq L_c$, denote by
	\begin{align*}
	h_L(\mPhi,\mPhi',\mG) := & f(\mPhi',\mG) + \langle \nabla_{\mPhi} f(\mPhi',\mG), \mPhi - \mPhi' \rangle\\& + \frac{L}{2}\|\mPhi - \mPhi'\|_F^2.
	\end{align*}
	Then, $f(\mPhi,\mG) \leq h_L(\mPhi,\mPhi',\mG)$ for all $(\mPhi,\mG), (\mPhi',\mG)\in \setL_{\rho_0}$.
	\label{lem:descent lemma}\end{lem}
The proof of  \cref{lem:descent lemma} is given in \ref{prf:descent lemma}. With \cref{lem:descent lemma}, we first establish that the sequence generated by  Algorithm \ref{alg: for G and Phi} is bounded and the limit point of any its convergent subsequence is a stationary point of $\rho$.
\begin{thm}[Subsequence convergence]
	Let $\{\mW_k =(\mPhi_k,\mG_k)\}_{k\geq 0}$ be the sequence generated by Algorithm \ref{alg: for G and Phi} with  step size $\eta< \frac{1}{L_c}$.
	Then the sequence $\{\mW_k\}$ is bounded and obeys the following properties:
	\begin{enumerate}[(P1)]
		\item sufficient decrease:
		\begin{align}
		&\rho(\mW_{k}) - \rho(\mPhi_{k+1},\mG_{k})\geq \frac{\frac{1}{\eta} - L_c}{2}\| \mPhi_{k}-\mPhi_{k+1}\|_F^2,\nonumber\\
		&\rho(\mPhi_{k+1},\mG_k) - \rho(\mW_{k+1})\geq \|\mG_k - \mG_{k+1}\|_F^2.
		\label{eq:sufficient decrease}\end{align}
		\item the sequence $\{\rho(\mPhi_k,\mG_k)\}_{k\geq 0}$ is convergent.
		\item convergent difference:
		\e
		\lim_{k\rightarrow \infty}\|\mW^{k+1}-\mW^k\|_F= 0.
		\label{eq:diff goes to 0}
		\ee
		\item for any convergent subsequence $\{\mW_{k'}\}$, its limit point $\underline\mW$ is a stationary point of $\rho$ and
		\e
		\lim_{k'\rightarrow \infty}\rho(\mW_{k'}) = \lim_{k\rightarrow \infty}\rho(\mW_{k}) = \rho(\underline\mW).
		\label{eq:lim f = f lim}\ee
	\end{enumerate}
	\label{thm:subsequence convergence}\end{thm}
The proof of Theorem \Cref{thm:subsequence convergence} is given in \ref{sec:prf thm subsequence convergence}. In a nutshell, Theorem \Cref{thm:subsequence convergence} implies that the sequence generated by Algorithm \ref{alg: for G and Phi} has at least one convergent subsequence, and the limit point of any convergent subsequence is a stationary point of $\rho$. The following result establishes that the sequence generated by  Algorithm \ref{alg: for G and Phi} is a Cauchy sequence and thus the sequence itself is convergent and converges to a stationary point of $\rho$. { Clearly, if the step size is chosen to satisfy \eqref{eq:sufficient decrease}, the convergence still holds. Thus, we suggest a backtracking method in \ref{Backtracking:unknow LC} to practically choose $\eta$.}

\begin{thm}[Sequence convergence] The sequence of iterates $\{(\mPhi_k,\mG_k)\}_{k\geq 0}$  generated by Algorithm \ref{alg: for G and Phi} with  step size $\eta< \frac{1}{L_c}$ converges to a stationary point of $\rho$.
	\label{thm:sequence convergence}\end{thm}
The proof of \cref{thm:sequence convergence} is given in \ref{sec:prf thm sequence convergence}. A special property named Kurdyka-Lojasiewicz (KL) inequality (see Definition~\ref{def:KL} in \ref{sec:prf thm sequence convergence}) of the objective function is introduced in proving Theorem \Cref{thm:sequence convergence}. We note that the KL inequality has been utilized to prove the convergence of proximal alternating minimization algorithms~\cite{attouch2010proximal,attouch2013convergence,bolte2014proximal}. Our proposed Algorithm \ref{alg: for G and Phi} differs from the proximal alternating minimization algorithms \cite{attouch2010proximal,attouch2013convergence,bolte2014proximal} in that we update $\mG$ (see \eqref{eq:problem G}) by exactly minimizing the objective function rather than utilizing a proximal operator (which decreases the objective function less than the one by exactly minimizing the objective function). Updating $\mG$ by exactly minimizing the objective function is popularly utilized in~\cite{elad2007optimized,duarte2009learning,abolghasemi2012gradient,li2013projection,li2015designing,hong2016efficient}. We believe our proof techniques for Theorems~\ref{thm:subsequence convergence} and \ref{thm:sequence convergence} will also be useful to analyze the convergence of other algorithms for designing sensing matrices~\cite{elad2007optimized,duarte2009learning,abolghasemi2012gradient,li2013projection,li2015designing,hong2016efficient}.

{Both Theorems \ref{thm:subsequence convergence} and \ref{thm:sequence convergence} hold for any fixed $\mPsi$, and hence any $\mA$ and $\bar \Psi$.} In terms of the step size for updating $\mPhi$, Algorithm \ref{alg: for G and Phi} utilizes a simple constant step size to simplify the analysis. But we note that the convergence analysis in Theorems \ref{thm:subsequence convergence} and \ref{thm:sequence convergence} can also be established for adaptive step sizes (such as obtained by the backtracking method), which may give faster convergence. 

When $\xi = 0$, $\calG_{\xi}$ consists of a single element (i.e., $\calG_{\xi} = \{\mId\}$) and the problem \eqref{eq:unconstrainded problem} is equivalent to
\e
\min_{\mPhi} \nu(\mPhi):= \|\mId - \mPsi^\T\mPhi^\T\mPhi\mPsi\|_F^2+\lambda\|\mPhi\|_F^2 + \delta_{\setS_\kappa}(\mPhi).
\label{eq:only Phi}\ee
Then, Algorithm~\ref{alg: for G and Phi} reduces to the projected gradient descent (PGD), which is known as the iterative hard thresholding (IHT) algorithm for compressive sensing \cite{blumensath2009iterative}. As a direct consequence of Theorems \ref{thm:subsequence convergence} and \ref{thm:sequence convergence}, the following result establishes convergence analysis of PGD for solving~\eqref{eq:only Phi}.
\begin{cor}\label{cor:PGD}
	Let $\{\mPhi_k\}_{k\geq 0}$ be the sequence generated by the PGD method with a constant step size $\eta< \frac{1}{L_c}$:
	\[
	\mPhi_{k+1} =\calP_{\calS_\kappa}(\mPhi_k- \eta \nabla_{\mPhi} f(\mPhi_{k},\mId)),
	\]
	where $\nabla_{\mPhi}f(\mPhi_k,\mId)$ is given in \eqref{eq:d Phi}.
	Then
	\begin{itemize}
		\item $\nu(\mPhi_k) - \nu(\mPhi_{k+1})\geq \frac{\frac{1}{\eta} - L_c}{2}\|\mPhi_k - \mPhi_{k+1}\|_F^2$.
		\item the sequence $\{\nu(\mPhi_k)\}_{k\geq 0}$ converges.
		\item the sequence $\{\mPhi_k\}$ converges to a stationary point of $\nu$.
	\end{itemize}
\end{cor}
We note that Corollary~\ref{cor:PGD} can also be established for PGD solving a general sparsity-constrained problem if the objective function is Lipschitz continuous. We end this section by
comparing Corollary~\ref{cor:PGD} with \cite[Theorem 3.1]{beck2013sparsity}, which provides convergence of PGD for solving a general sparsity-constrained problem. Corollary~\ref{cor:PGD} reveals that the sequence generated by PGD is convergent and converges to a stationary point, while \cite[Theorem 3.1]{beck2013sparsity} only shows subsequential convergence property of PGD, i.e., the limit point of any convergent subsequence converges to a stationary point.

{We end this section by noting that an alternative approach is to pose the sparsity for the entire sensing matrix instead of each row.
	\Cref{alg: for G and Phi} can be directly utilized for designing such sparse sensing matrix by simply revising the projection operator in updating the sensing matrix $\bm \Phi$. But we empirically observe that such sensing matrix has slightly inferior performance than the one obtained by imposing sparsity on each row. Also, the reason that we do not impose sparsity in each column is because $M$ is usually small as $M<< N$,  largely restricting the sparsity level of the sensing matrix $\mPhi$. For example, when $M = 10$ and $N = 100$ and if we want to design a sensing matrix with only $10\%$ nonzero elements. Then if we impose the sparsity to each column, then each column can only have one nonzero element which is not easy to optimize with, whereas each row can have ten nonzero elements if we impose the sparsity on each row.}

\section{Simulations}
\label{sec:simulations}
A set of experiments on synthetic data and real images are conducted in this section to illustrate the performance of the proposed method for designing sparse sensing matrix. We compare with several existing methods for designing sensing matrices  \cite{mamaghanian2011compressed,li2013projection,hong2017SP}. For a given dictionary $\mPsi$, different sensing matrices resulting in various CS systems, we list below all possible CS systems that are utilized in this paper. 
%We initialize Algorithm \ref{alg: for G and Phi} with the sensing matrix given by \cite[Theorem 2]{li2013projection}, but one can get similar performance with a random initialization with more iterations.

\vspace{.1in}
\begin{mdframed}[backgroundcolor=black!3,rightline=false,leftline=false,    leftmargin =0pt,
	rightmargin =0pt,
	innerleftmargin =2pt,
	innerrightmargin =2pt]
	{\small
		\begin{tabular}{ll}
			CS$_{randn}$: & $\mPsi~+$ A dense random matrix \\
			%CS$_{robsensing}$:  & $\mPsi~+$ Sensing matrix obtained in \cite{li2015designing}\\
			CS$_{MT}$:  & $\mPsi~+$  Sensing matrix \cite{hong2017SP}\\
			CS$_{MT-ETF}$:  & $\mPsi~+$ Sensing matrix \cite{hong2017SP} \\
			CS$_{LZYCB}$:  & $\mPsi~+$ Sensing matrix \cite{li2013projection}\\
			CS$_{bispar}$:  & $\mPsi~+$   A binary sparse sensing matrix \cite{mamaghanian2011compressed}\\
			CS$_{sparse-A}$: & $\mPsi~+$  Output of {\bf Algorithm \ref{alg: for G and Phi}} with \\& ~~~~~~~$\xi = 0$ (i.e., $\setG_\xi = \{\mId\}$) and $\mA=\text{DCT}$ \\
			CS$_{sparse}$:  & $\mPsi~+$ Output of {\bf Algorithm \ref{alg: for G and Phi}} with\\ &~~~~~~~$\xi = 0$ (i.e., $\setG_\xi = \{\mId\}$) and $\mA = \mId$  \\
			%            CS$_{sparse-AETF}$:  & $\mPsi~+$ Output of {\bf Algorithm 1} with\\ &~~~~~~~$\xi = \underline{\mu}$ and $\mA = \text{DCT}$\\
			CS$_{sparse-ETF}$:  & $\mPsi~+$ Output of {\bf Algorithm \ref{alg: for G and Phi}} with\\ &~~~~~~~$\xi = \underline{\mu}$ and $\mA = \mId$
		\end{tabular}
	}
\end{mdframed}

\subsection{Synthetic Data}

We generate an $N\times L$ dictionary  $\mPsi$ with normally
distributed entries and an $M\times N$  random matrix $\mPhi_0$ for CS$_{randn}$. The training and testing data are built as follows: with the given dictionary $\mPsi$, generate a set of $J$
$K$-sparse vectors $\{\vs_i\in\R^{L}\}_{i=1}^J$, where the index of the non-zero elements in $\vs_i$ obeys a normal distribution; then obtain the sparse signals  $\{\vx_i\}_{i=1}^J$ through \e
\vx_i=\mPsi \vs_i + \ve_i,
\label{eq:generate signal}\ee
where $\ve_i$ denotes the Gaussian noise with mean zero and covariance $\sigma^2$. Denote SNR as the signal-to-noise ratio (in dB) of the signals in \eqref{eq:generate signal}.

The performance of a CS system is evaluated via the mean squared error (MSE)
\e
\text{MSE}\triangleq \frac{1}{N\times J}\sum_{i=1}^J\|\vx_i - \hat\vx_i\|^2_2,\label{MSE}
\ee
where $\hat{\vx}_i = \mPsi \hat{\vs}_i$ denotes the recovered signal and $\hat{\vs}_i$ is obtained through
\[
\hat{\vs}_i = {\argmin}_{\vs_i}|| \mPhi\vx_i -\mPhi\mPsi \vs_i||^2_2~~~\text{s.t.}~~||\vs_i||_0\leq K,~\forall i.
\]
%with $\vz_i =$.
%by OMP.

Now, we examine the convergence of Algorithm \ref{alg: for G and Phi}. \Cref{fig:ETF} shows the objective function value $f(\mPhi_k,\mG_k)$ and the values of $\|\mPhi_{k+1} -\mPhi_{k}\|_F$ and $\|\mG_{k+1} -\mG_{k}\|_F$ versus number of iteration. We see $f(\mPhi_k,\mG_k)$ decays steadily and $\|\mPhi_{k+1} -\mPhi_{k}\|_F$ and $\|\mG_{k+1} -\mG_{k}\|_F$ decrease to 0 linearly. This coincides with our theoretical  analysis. 

{Next, we discuss the choice of $\lambda$. As we mentioned in Section \ref{sec:preliminary}, $\lambda$ is used to balance the importance of mutual coherence and the robustness of SRE. In \Cref{figure:chooselambda}, we show the optimal value of $\lambda$ for CS$_{sparse}$ with different SNR.\footnote{The optimal $\lambda$ means the corresponding sensing matrix yielding a highest recovery accuracy.} We observe that the optimal $\lambda$ becomes large when the SNR is low coinciding with our expectation.}

\begin{figure}[!htb]
	\centering
	\includegraphics[height=3.5cm]{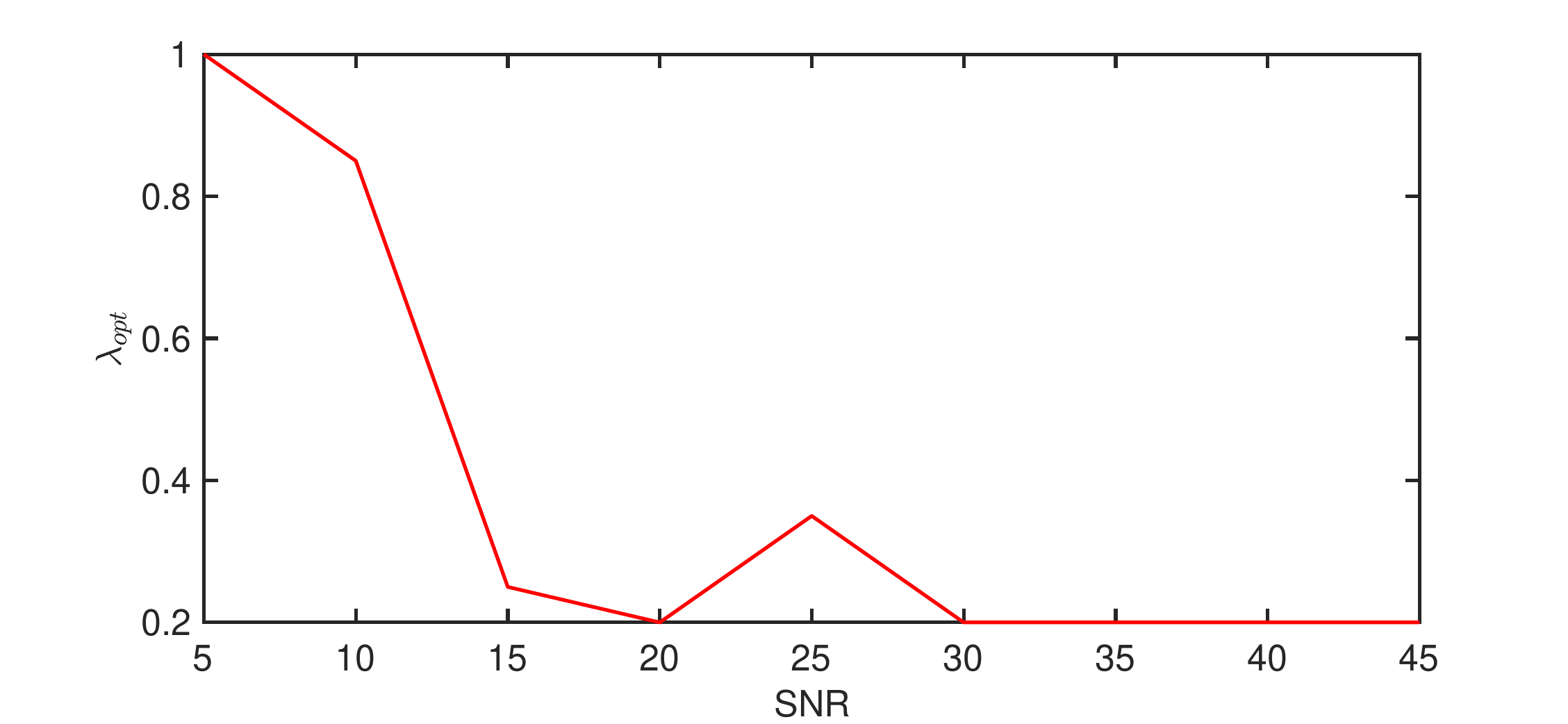}
	\caption{The value of optimal $\lambda$ versus varying SNR for CS$_{sparse}$.  Here, we set $M = 25, N = 60, L = 80,K=4, J = 2000$, and $\kappa =20$.}\label{figure:chooselambda}
\end{figure}

We then compare the performance of each CS system with varying SNR level in \eqref{eq:generate signal}. The results are displayed in \Cref{fig5}.\footnote{For synthetical experiments, we set the base matrix $\mA$ as an identity matrix. We will exploit the performance of adding the DCT matrix as the base matrix for natural images.}  
This experiment indicates  CS$_{MT}$, CS$_{MT-ETF}$, CS$_{sparse}$ and CS$_{sparse-ETF}$ outperform the others when SNR $<25$dB.  We also see CS$_{MT-ETF}$ and CS$_{sparse-ETF}$ outperforms CS$_{MT}$ and CS$_{sparse}$ when SNR is high, respectively. Despite the high performance of  CS$_{LYZCB}$ when SNR $>25$dB, it decays fast as SNR decreases, which reveals CS$_{LYZCB}$ is not robust to SRE.  Interestingly, the corresponding sparse sensing matrices CS$_{sparse}$ and CS$_{sparse-ETF}$ have comparable performance as CS$_{MT}$ and CS$_{MT-ETF}$, and are much better than CS$_{randn}$ and CS$_{bispar}$.

%We also demonstrate the performance of the CS systems while changing $M $ from 5 to 60 in Fig. \ref{fig:varyM}.  All CS systems improves performance when $M$ becomes larger and CS$_{LYZCB}$, CS$_{MT}$, CS$_{MT-ETF}$, CS$_{sparse}$ and CS$_{sparse-ETF}$ stop decreasing the MSE when $M>30$. And CS$_{sparse}$ and CS$_{sparse-ETF}$ with sparse sensing matrices whose sparsity $\kappa =20$ have comparable performance to the dense ones CS$_{MT}$, CS$_{MT-ETF}$, and outperform CS$_{randn}$ and CS$_{bispar}$ for all $M$. 

\begin{figure}[!htb]    
	\subfigure[]{ \includegraphics[scale=0.26]{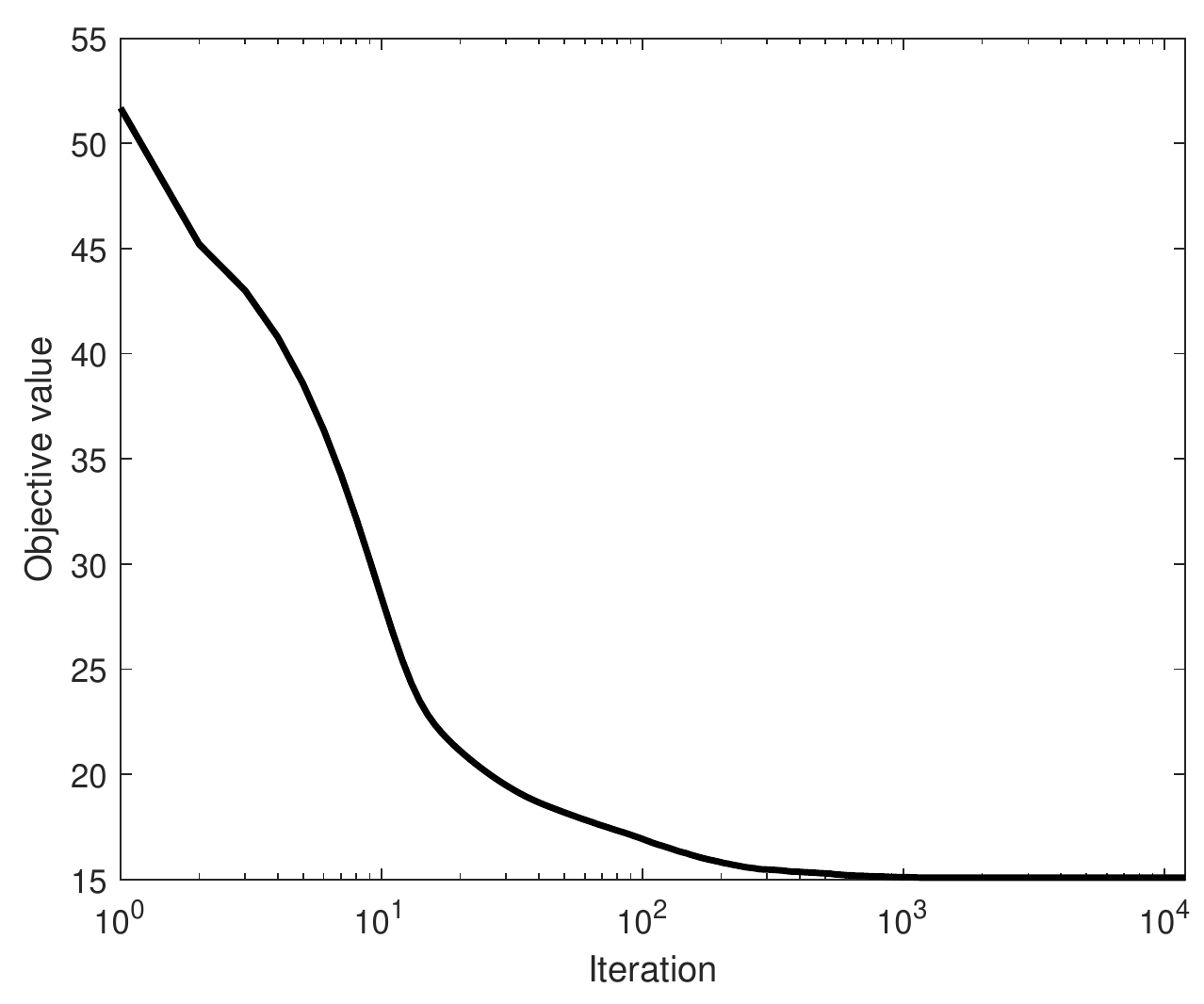}}
	\subfigure[]{ \includegraphics[scale=0.26]{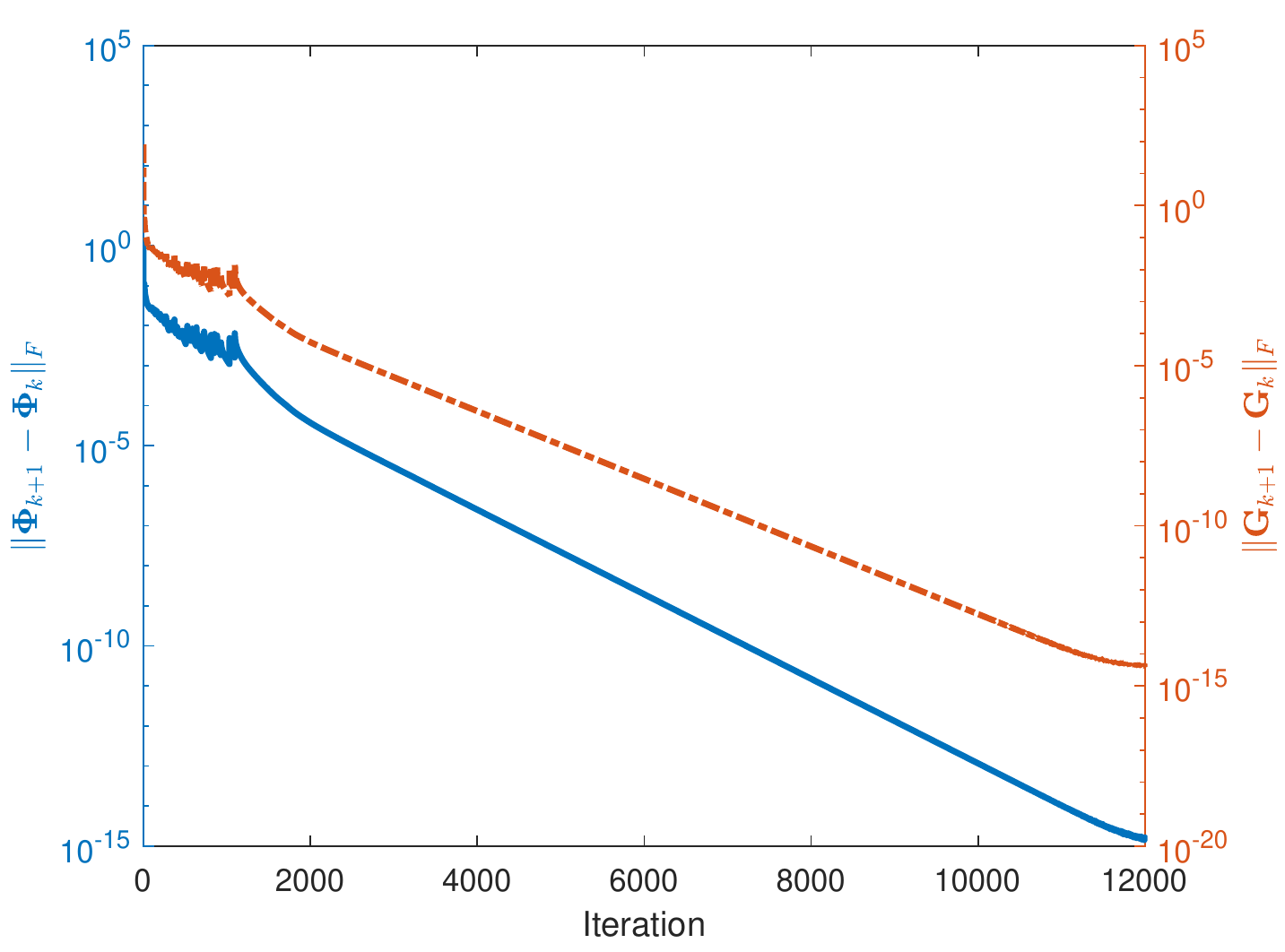}}
	\caption{Convergence of Algorithm \ref{alg: for G and Phi} {in terms of (a) objective value, (b) the change  of iterates $\{\mPhi\}_k$ (blue line) and $\{\mG\}_k$ (red line)}.  Here, we set $M = 25, N = 60, L = 80, \lambda  = 0.25$ and $\kappa =20$.}\label{fig:ETF}
\end{figure}

\begin{figure}[!htb]
	\centering
	\includegraphics[height=3.5cm]{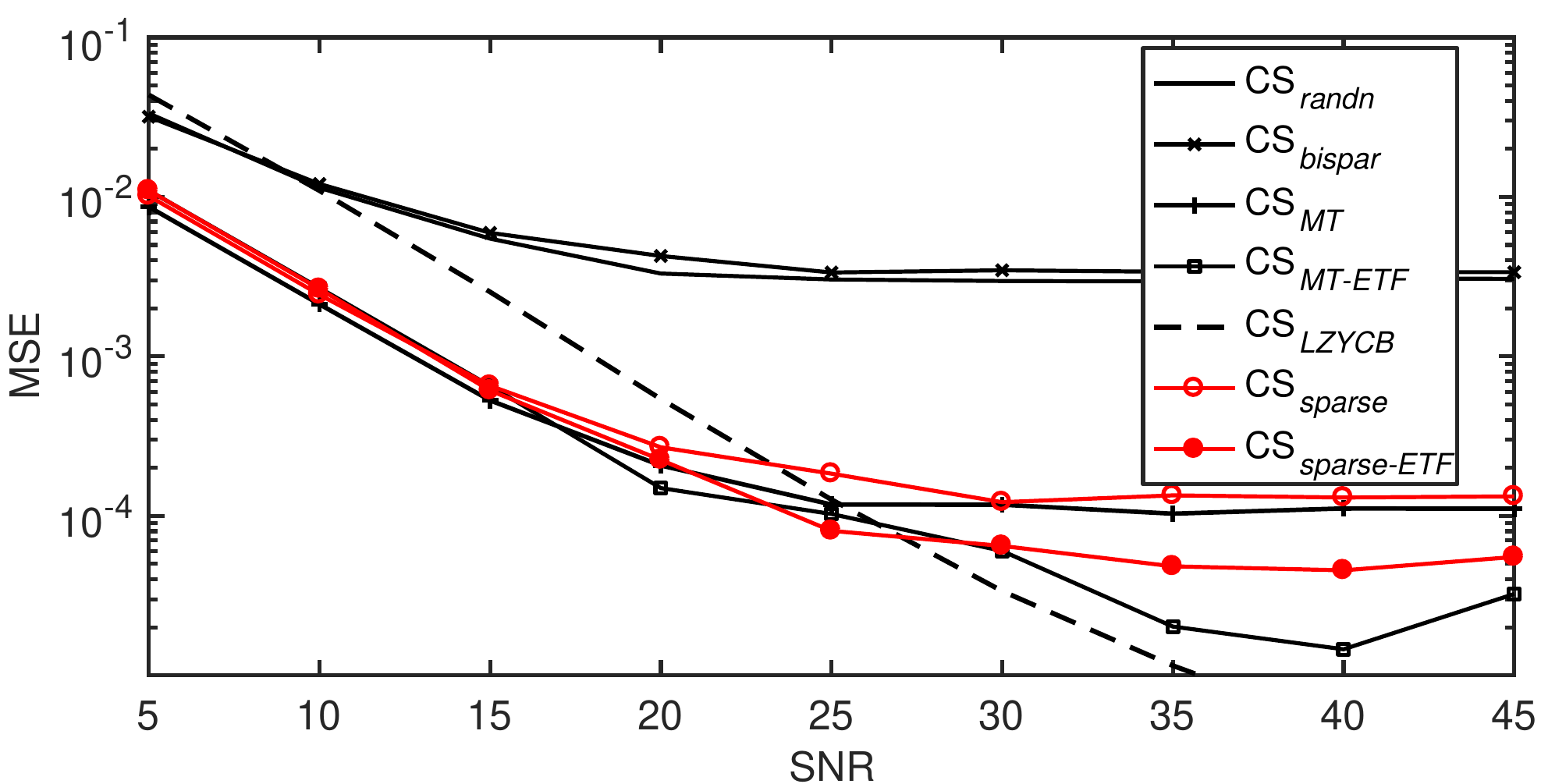}
	\caption{ MSE versus different SNR (dB) for the signals in \eqref{eq:generate signal}. Here, we set $M = 25, N = 60, L = 80,K=4, J = 2000,  \lambda  = 0.25$ and $\kappa =20$. Disappearance from this figure means MSE is less than $10^{-5}$. }\label{fig5}
\end{figure}

{Finally, we investigate the change of signal recovery accuracy with varying $M$ and $K$. In \Cref{figure:varyM,figure:VaryKSNR20}, we show the recovery accuracy versus various $M$ and $K$. We observe that our approach almost works the same as the dense one, CS$_{MT}$ and CS$_{MT-ETF}$.}

\begin{figure}[!htb]
	\centering
	\includegraphics[height=3.5cm]{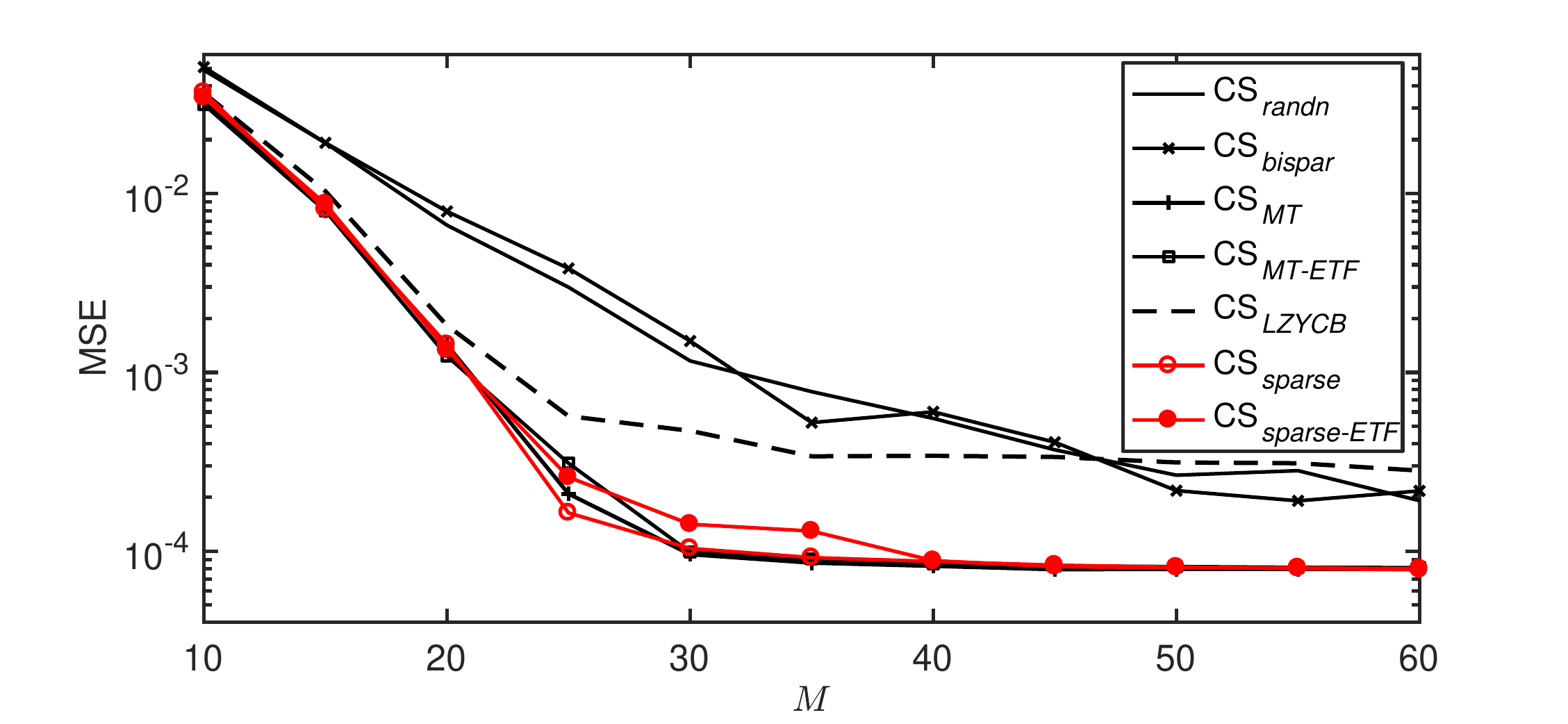}
	\caption{MSE versus different $M$ with SNR$=20$dB. Here, we set $N = 60, L = 80, K=4,J = 2000,  \lambda  = 0.25$ and $\kappa =20$.}\label{figure:varyM}
\end{figure}

\begin{figure}[!htb]
	\centering
	\includegraphics[height=3.5cm]{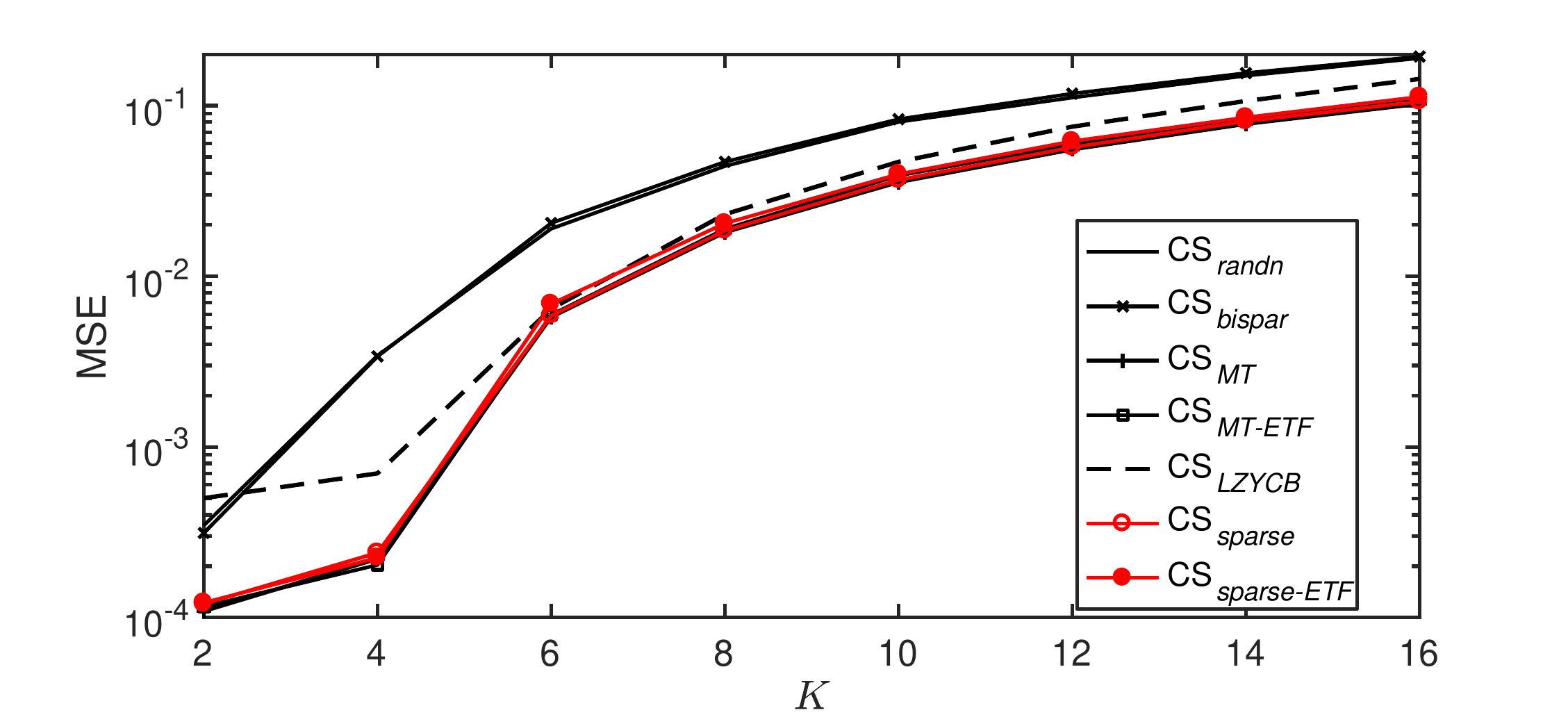}
	\caption{MSE versus different $K$ with SNR$=20$dB. Here, we set $M = 25, N = 60, L = 80, J = 2000,  \lambda  = 0.25$ and $\kappa =20$.}\label{figure:VaryKSNR20}
\end{figure}

%\begin{figure}[!htb]
%    \includegraphics[height=4cm]{Figures/synthetic_varyM.pdf}
%    \caption{ MSE versus sensing dimension $M$. Here, we set $N = 60, L = 80,K=4, J = 2000,  \lambda  = 0.25$  SNR $=25$ dB and $\kappa =20$.  }\label{fig:varyM}
%\end{figure}

\subsection{Real Images}

We now apply these CS systems to real image reconstruction from their sensing measurements.
%\footnote{ \cite{li2013projection,li2015designing} presents that the ETF-based sensing matrix design method (i.e. optimize $\mG$ in \eqref{eq:existing approaches} too)  outperforms  the sensing matrix obtained by setting $\mG =\mId$ in \eqref{eq:existing approaches} for exactly sparse signal, but has slightly worse performance for signals (like real images) that are not exactly sparse.  This also coincides our observation in \Cref{fig5} for CS$_{LZYCB}$, CS$_{MT-ETF}$ and CS$_{sparse-ETF}$.  Thus we only show the results for CS$_{MT}$ and CS$_{sparse}$  to save space.}. 
We examine the performance of the CS systems with two different sizes of dictionaries, a low dimensional dictionary  $\mPsi\in\Re^{64\times 100}$ and a high dimensional one $\mPsi\in\Re^{256\times 800}$. 

The low dimensional  dictionary $\mPsi$ is learnt through KSVD algorithm \cite{aharon2006ksvd} with a set of $\sqrt{N}\times \sqrt{N}$ non-overlapping patches by extracting randomly 15 patches from each of 400 images in the LabelMe \cite{russell2008labelme} training data set. With each patch of $\sqrt{N}\times \sqrt{N}$ re-arranged as a vector of $N\times 1$, a set of $N\times 6000$ signals are obtained. Similarly, for learning a high dimensional dictionary, we extract more patches from the training dataset and obtains nearly $10^6$ signals {{since a high dimensional dictionary has much more parameters that need to be trained than a low dimensional dictionary}}. To address such a large training dataset, we choose the online dictionary learning algorithm \cite{mairal2009online}  to learn this dictionary. 

%$N$ is set to be $64$ and $256$ for the low and high dimensional dictionaries, respectively.

The performance of each CS system for real images is evaluated through Peak Signal to Noise Ratio (PSNR),
$$ 
\text{PSNR} \triangleq  10\times \log{10 \left[\frac{(2^r-1)^2}{\text{MSE}}\right]} \text{(dB)} 
$$
where $r=8$ bits per pixel and MSE is defined in \eqref{MSE}.

We choose several test images to demonstrate the reconstruction performance in terms of PSNR. Since patch-based processing of images will introduce the artifact on boundary called blockiness, the deblocking techniques can be introduced here to act as a post-processing step to reduce such an artifact. To this end, we utilize the BM3D denoising algorithm as post-processing to tackle the blockiness \cite{dabov2007image}. We observe that such a post-processing step not only improves the visual effect, but also increases the PSNR for each method. To illustrate the improvement of the PSNR, we list the amount of increased PSNR by the post-processing in Tables \ref{table:PSNR low dim} and \ref{table:PSNR high dim} and \Cref{fig:images2}. 
%The corresponding results are shown in Tables \ref{table:PSNR low dim} and \ref{table:PSNR high dim}.

\begin{table*}[htb!]
	\footnotesize
	\caption{Statistics of PSNR (dB) for six images with $M=20,~N=64,~L=100,~K=4,~\lambda = 1.4$ using different $\kappa$. The performance of each CS system is described in two rows: the first row is the PSNR without post processing, while the second row is the amount of improved PSNR by post processing.} \label{table:PSNR low dim}
	\vspace{-.1in}
	\begin{center}
		\begin{tabular}{c||c|c||c|c||c|c||c|c||c|c||c|c}
			\hline\hline
			\multirow{2}{*}{} &
			
			\multicolumn{2}{c||}{Lena} &
			\multicolumn{2}{|c||}{Couple} &
			\multicolumn{2}{|c||}{Barbara} &
			\multicolumn{2}{|c||}{Child} &
			\multicolumn{2}{|c||}{Plane} &
			\multicolumn{2}{c}{Man} \\
			\cline{2-3} \cline{4-5} \cline{6-7}\cline{8-9}  \cline{10-11}  \cline{12-13}
			& $\kappa=10$    &$\kappa=20$      & $\kappa=10$     &$\kappa=20$   & $\kappa=10$     &$\kappa=20$  & $\kappa=10$     &$\kappa=20$     & $\kappa=10$     &$\kappa=20$     & $\kappa=10$     &$\kappa=20$    \\
			
			\hline
			\multirow{2}{*}{ CS$_{randn}$ }      &\multicolumn{2}{c||}{$29.69$}   &\multicolumn{2}{c||}{ $27.01$} &\multicolumn{2}{c||}{$22.44$}    &\multicolumn{2}{c||}{$31.20$}  &\multicolumn{2}{c||}{$28.57$}   &\multicolumn{2}{c}{ $27.41$} \\
			
			%\cline{2-13} 
			
			&\multicolumn{2}{c||}{$+1.12$}   &\multicolumn{2}{c||}{ $+0.95$} &\multicolumn{2}{c||}{$+0.52$}    &\multicolumn{2}{c||}{$+1.27$}  &\multicolumn{2}{c||}{$+1.08$}   &\multicolumn{2}{c}{ $+1.04$}\\
			% \hline
			%{ CS$_{robsensing}$}   
			%  &\multicolumn{2}{c||}{$32.66$}    &\multicolumn{2}{c||}{$29.99$}   &\multicolumn{2}{c||}{$25.22$}    &\multicolumn{2}{c||}{$34.18$}  &\multicolumn{2}{c||}{ $31.57$}   &\multicolumn{2}{c}{$30.26$ }\\
			
			\hline
			
			\multirow{2}{*}{ CS$_{MT}$ } &\multicolumn{2}{c||}{ $32.75$}    &\multicolumn{2}{c||}{$30.01$}   &\multicolumn{2}{c||}{$25.36$}  &\multicolumn{2}{c||}{$34.22$}   &\multicolumn{2}{c||}{$31.60$}   &\multicolumn{2}{c}{ $30.39$ }\\

			&\multicolumn{2}{c||}{ $+1.31$}    &\multicolumn{2}{c||}{$+1.05$}   &\multicolumn{2}{c||}{$+0.45$}  &\multicolumn{2}{c||}{$+1.72$}   &\multicolumn{2}{c||}{$+0.78$}   &\multicolumn{2}{c}{ $+0.32$ }\\

			\hline
			
			\multirow{2}{*}{  CS$_{LZYCB}$  }
			&\multicolumn{2}{c||}{$12.74$ } &\multicolumn{2}{c||}{$10.38$  }  
			&\multicolumn{2}{c||}{$~4.19$}   &\multicolumn{2}{c||}{$15.93$} 
			&\multicolumn{2}{c||}{$14.51$ }   &\multicolumn{2}{c}{$~9.75$}  \\

			&\multicolumn{2}{c||}{$+1.58$ } &\multicolumn{2}{c||}{$+2.08$  }  
			&\multicolumn{2}{c||}{$+7.32$}   &\multicolumn{2}{c||}{$+1.76$} 
			&\multicolumn{2}{c||}{$+1.33$ }   &\multicolumn{2}{c}{$+3.20$}  \\

			\hline
			\multirow{2}{*}{ CS$_{bispar}$ }     &$29.36$  &$29.27$    &$26.85$   &$26.87$   &$22.42$  &$22.49$   &$30.80$   &$30.87$  &$28.13$  &$28.28$   &$27.19$   &$27.14$ \\
			
			&$+1.23$  &$+1.24$    &$+0.97$   &$+1.03$   &$+0.65$  &$+0.61$   &$+1.56$   &$+1.46$  &$+0.78$  & $+0.90$  &$+1.07$  &$+1.09$ \\
			
			\hline
			
			\multirow{2}{*}{ CS$_{sparse-A}$}  &$32.38$  &$32.65$    &$29.63$   &$29.88$   &$24.91$ &$25.19$    &$33.84$   &$34.12$  &$31.28$  &$31.52$   &$30.01$  &$30.27$\\

			&$+1.86$  &$+1.34$    &$+1.31$   &$+1.07$   &$+1.04$ &$+0.47$    &$+2.03$   &$+1.70$  & $+0.85$ & $+1.64$  &$+0.44$  & $+0.38$\\
			\hline
			\multirow{2}{*}{ CS$_{sparse}$ }     &$32.26$  &$32.56$    &$29.47$   &$29.79$   &$24.76$ &$25.11$    &$33.75$   &$34.07$  &$31.15$  &$31.48$   &$29.83$  &$30.15$ \\
			
			&$+1.26$  &$+1.33$    &$+1.01$   &$+1.02$   &$+0.54$ &$+0.47$    &$+1.71$   &$+1.68$  &$+0.86$   & $+1.21$ &$+0.61$  & $+0.40$ \\
			\hline \hline
			
		\end{tabular}
	\end{center}
\end{table*}

\begin{table*}[htb!]
	\footnotesize
	\caption{Similar to \Cref{table:PSNR low dim}, but with  $M=80,~N=256,~L=800,~K = 16,~\lambda = 0.5$.}\label{table:PSNR high dim}
	\vspace{-.1in}
	\begin{center}
		\begin{tabular}{c||c|c|| c|c ||c|c ||c|c ||c|c || c|c}
			\hline\hline
			\multirow{2}{*}{} &
			\multicolumn{2}{c||}{Lena} &
			\multicolumn{2}{|c||}{Couple} &
			\multicolumn{2}{|c||}{Barbara} &
			\multicolumn{2}{|c||}{Child} &
			\multicolumn{2}{|c||}{Plane} &
			\multicolumn{2}{c}{Man} \\
			
			\cline{2-3}       \cline{4-5}  \cline{6-7} \cline{8-9}  \cline{10-11}  \cline{12-13}
			& $\kappa=10$  &$\kappa=30$   & $\kappa=10$   &$\kappa=30$  & $\kappa=10$     &$\kappa=30$  & $\kappa=10$     &$\kappa=30$     & $\kappa=10$     &$\kappa=30$     & $\kappa=10$     &$\kappa=30$    \\
			
			\hline
			\multirow{2}{*}{ CS$_{randn}$  }     &\multicolumn{2}{c||}{$30.73$}    &\multicolumn{2}{c||}{$27.40$}  &  \multicolumn{2}{c||}{$22.76$}  &\multicolumn{2}{c||}{$31.99$}  &\multicolumn{2}{c||}{$29.57$ }   &\multicolumn{2}{c}{$27.94$} \\
			
			&\multicolumn{2}{c||}{$+0.69$}    &\multicolumn{2}{c||}{$+0.51$}  &  \multicolumn{2}{c||}{$+0.21$}  &\multicolumn{2}{c||}{$+0.59$}  &\multicolumn{2}{c||}{$+0.48$ }   &\multicolumn{2}{c}{$+0.48$} \\
			
			\hline
			\multirow{2}{*}{CS$_{MT}$}           &\multicolumn{2}{c||}{$34.38$}    &\multicolumn{2}{c||}{$30.90$}   &\multicolumn{2}{c||}{ $26.04$}    &\multicolumn{2}{c||}{$35.60$} &\multicolumn{2}{c||}{$33.34$ }  &\multicolumn{2}{c}{$31.50$ } \\
			
			&\multicolumn{2}{c||}{$+0.14$}    &\multicolumn{2}{c||}{$+0.23$}   &\multicolumn{2}{c||}{ $+0.13$}    &\multicolumn{2}{c||}{$+0$} &\multicolumn{2}{c||}{$+0.24$ }  &\multicolumn{2}{c}{$+0.20$ } \\

			\hline
			\multirow{2}{*}{CS$_{bispar}$}     &$30.23$  &$30.72$    &$27.33$   &$27.36$   &$22.57$  &$22.64$    &$31.70$   &$31.94$  &$29.29$  &$29.49$   &$27.59$  &$27.82$ \\
			
			&$+0.67$  &$+0.67$    &$+0.50$   &$+0.52$   &$+0.21$  &$+0.20$    &$+0.58$   &$+0.59$  &$+0.48$  &$+0.48$   &$+0.46$  &$+0.47$ \\

			\hline
			\multirow{2}{*}{CS$_{sparse-A}$}   &$33.89$  &$34.24$    &$30.47$   &$30.75$   &$25.44$  &$25.82$    &$35.09$   &$35.36$  &$32.92$  &$33.15$   &$30.95$  &$31.22$\\
			
			&$+0.34$  &$+0.22$    &$+0.37$   &$+0.29$   &$+0.18$  &$+0.16$    &$+0.25$   &$+0.08$  &$+0.38$  &$+0.29$   &$+0.35$  &$+0.27$\\
			
			\hline
			\multirow{2}{*}{CS$_{sparse}$}  &$33.17$  &$33.50$    &$29.69$   &$29.94$   &$24.26$  &$24.60$    &$34.38$   &$34.66$  &$32.41$  &$32.61$   &$30.06$  &$30.39$ \\
			
			&$+0.45$  &$+0.43$    &$+0.42$   &$+0.40$   &$+0.15$  &$+0.16$    &$+0.24$   &$+0.21$  &$+0.38$  &$+0.38$   &$+0.35$  &$+0.36$ \\
			
			\hline \hline
		\end{tabular}
	\end{center}
\end{table*}

With image €œLena, we show the PSNR versus the sparsity $\kappa$ (the number of non-zero elements in each row of the sparse sensing matrix) in \Cref{fig:PSNR high dim}.  And furthermore, we list the performance statistics on other images {including Couple, Barbara, Child, Plane, and Man} in Tables \ref{table:PSNR low dim} and \ref{table:PSNR high dim}. \Cref{fig:images2}  displays the visual result of ``couple''.

As expected, CS$_{sparse-A}$ and CS$_{sparse}$ yield higher PSNR when increases the sparsity $\kappa$. It is interesting to see that although the sparsity is very low, for example, $\kappa = 10$, CS$_{sparse-A}$ is only 0.53dB inferior to CS$_{MT}$ and still has more than $3$dB better than CS$_{randn}$ which is a dense sensing matrix.
We note that the gap between CS$_{sparse-A}$ and CS$_{MT}$ is almost negligible (with 0.15 dB) when $\kappa \geq 30$. This meets our argument that we can design a sparse sensing matrix instead of a dense matrix resulting in similar performance so that we can reduce the computational cost for sensing signals.  

\Cref{fig:PSNR high dim} and \Cref{table:PSNR low dim,table:PSNR high dim} indicate  CS$_{sparse-A}$ has better performance than CS$_{sparse}$ because of $\mA$  when $\kappa$ is small. This demonstrates the effect of utilizing the auxiliary DCT matrix for designing a structured sparse sensing matrix to increase the reconstruction accuracy. %However, we note that the existence of $\mA$ also increases the sensing cost  in CS$_{sparse-A}$  (up to $O(N\log N)$) than CS$_{sparse}$.

It is not surprising to note that CS$_{LZYCB}$ yields very low PSNR for real image experiments. This coincides with \Cref{fig5} and further demonstrates that the sensing matrix in CS$_{LZYCB}$ is not robust to the SRE. But we observe that the proposed sparse sensing matrices are robust to the SRE and hence the CS$_{sparse}$ and CS$_{sparse-A}$ have higher PSNR.

Comparing the results in Table~\ref{table:PSNR low dim} and Table~\ref{table:PSNR high dim},
we observe that with a high dimensional dictionary, higher PSNR can be obtained. It is of great interest to note that the proposed sparse sensing matrix becomes extremely efficient for high dimensional patches since it can significantly reduce the sensing costs.

% We finally note that we only compare with CS$_{LYZCB}$ and CS$_{MT}$ because as shown in \cite{hong2017SP} (and we also observed this), CS$_{MT}$ has better performance than the ones in \cite{elad2007optimized,abolghasemi2012gradient,li2015designing} for real image reconstruction.

\begin{figure}[!htb]
	\centering
	\includegraphics[height=3.5cm]{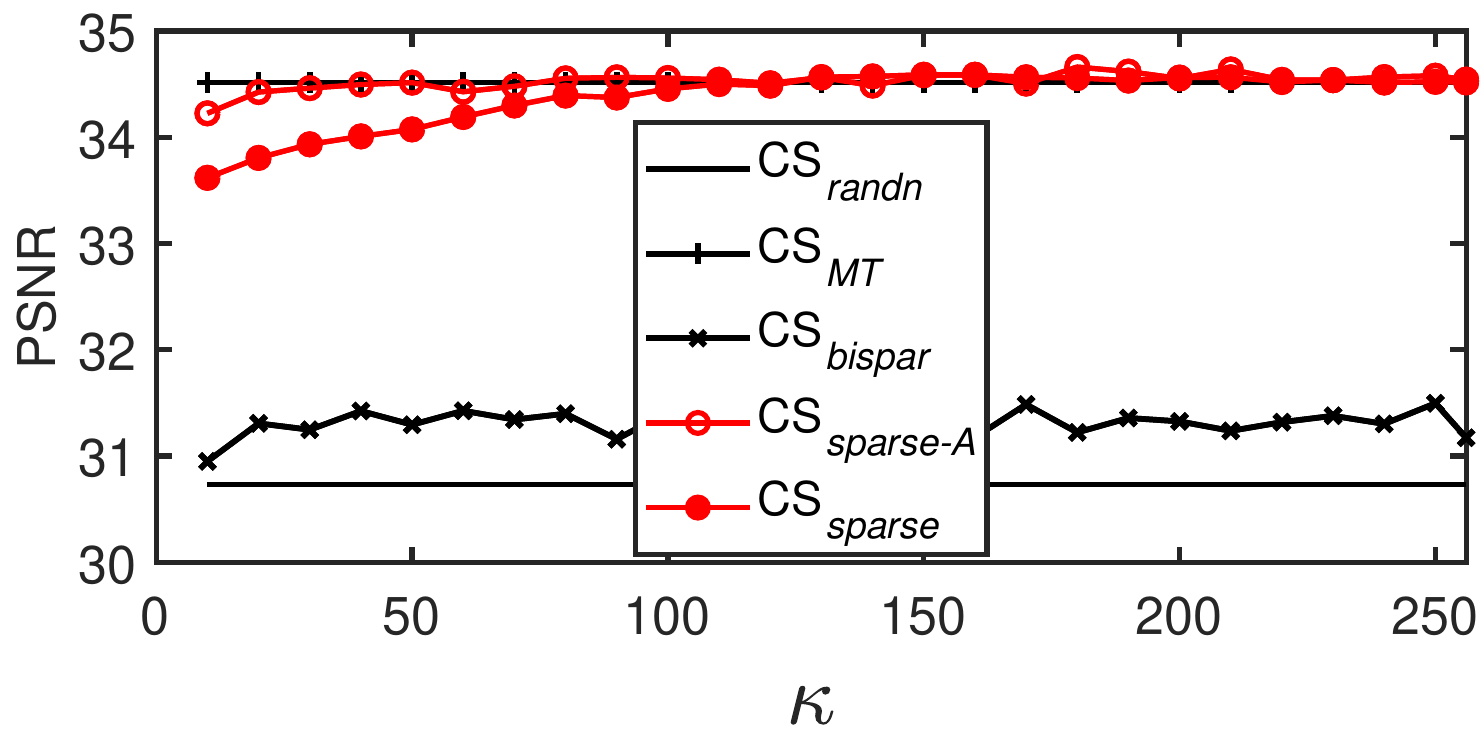}
	\caption{PSNR(dB) versus sparsity $\kappa$ for different CS systems { without post-processing} on image Lena. Here, $M = 80, N = 256, L = 800, K=16,\lambda  = 0.5$. }\label{fig:PSNR high dim}
\end{figure}

\begin{figure}
	\centering
	\includegraphics[width=0.25\textwidth]{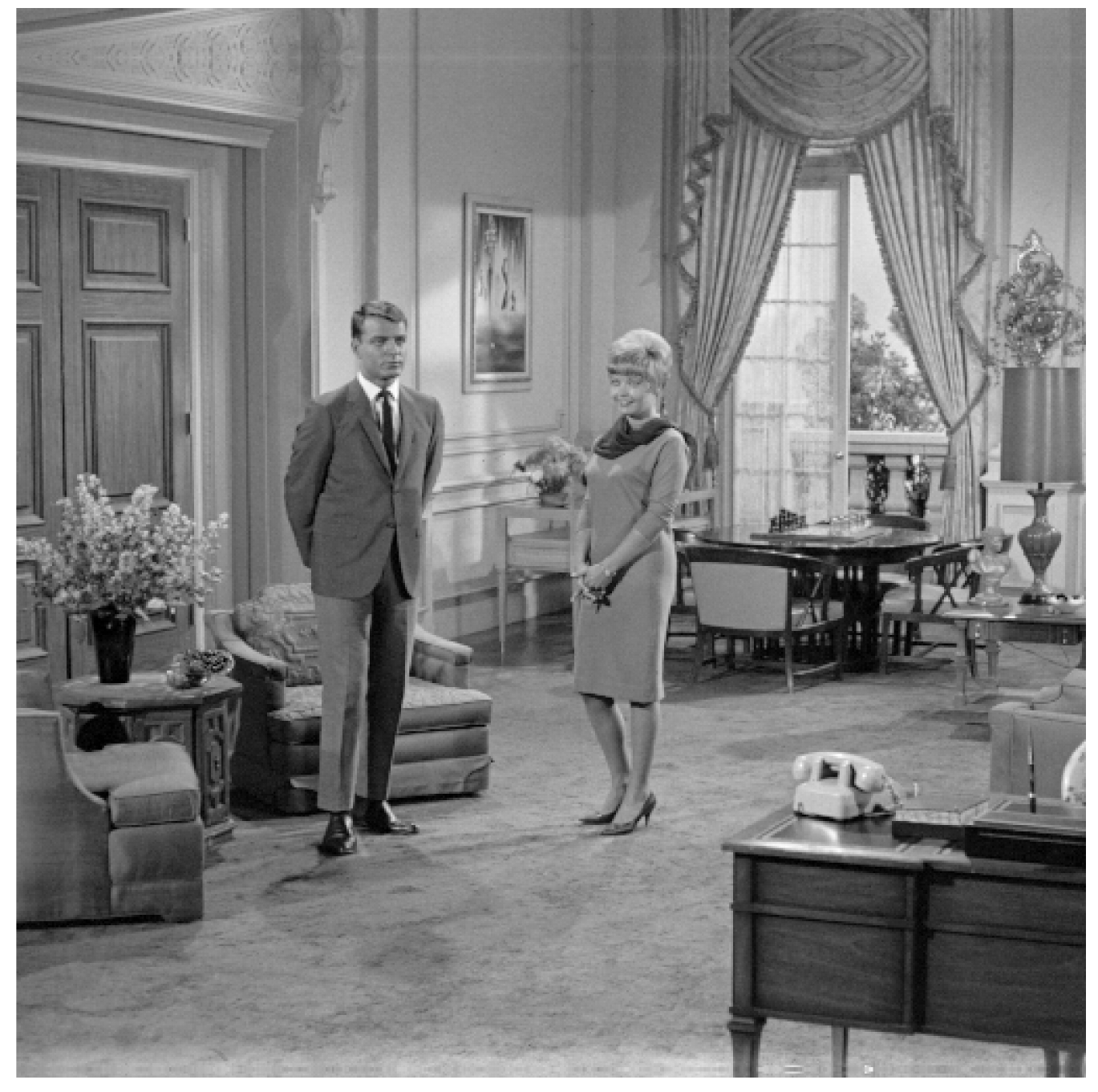}
	\caption{The Original test image: ``Couple''.}
	\label{fig:original:couple}
\end{figure}

\begin{figure*}
	\centering
	
	\vspace{-0.2cm}
	\subfigure[CS$_{randn}$: $27.40$]{\includegraphics[width=0.18\textwidth]{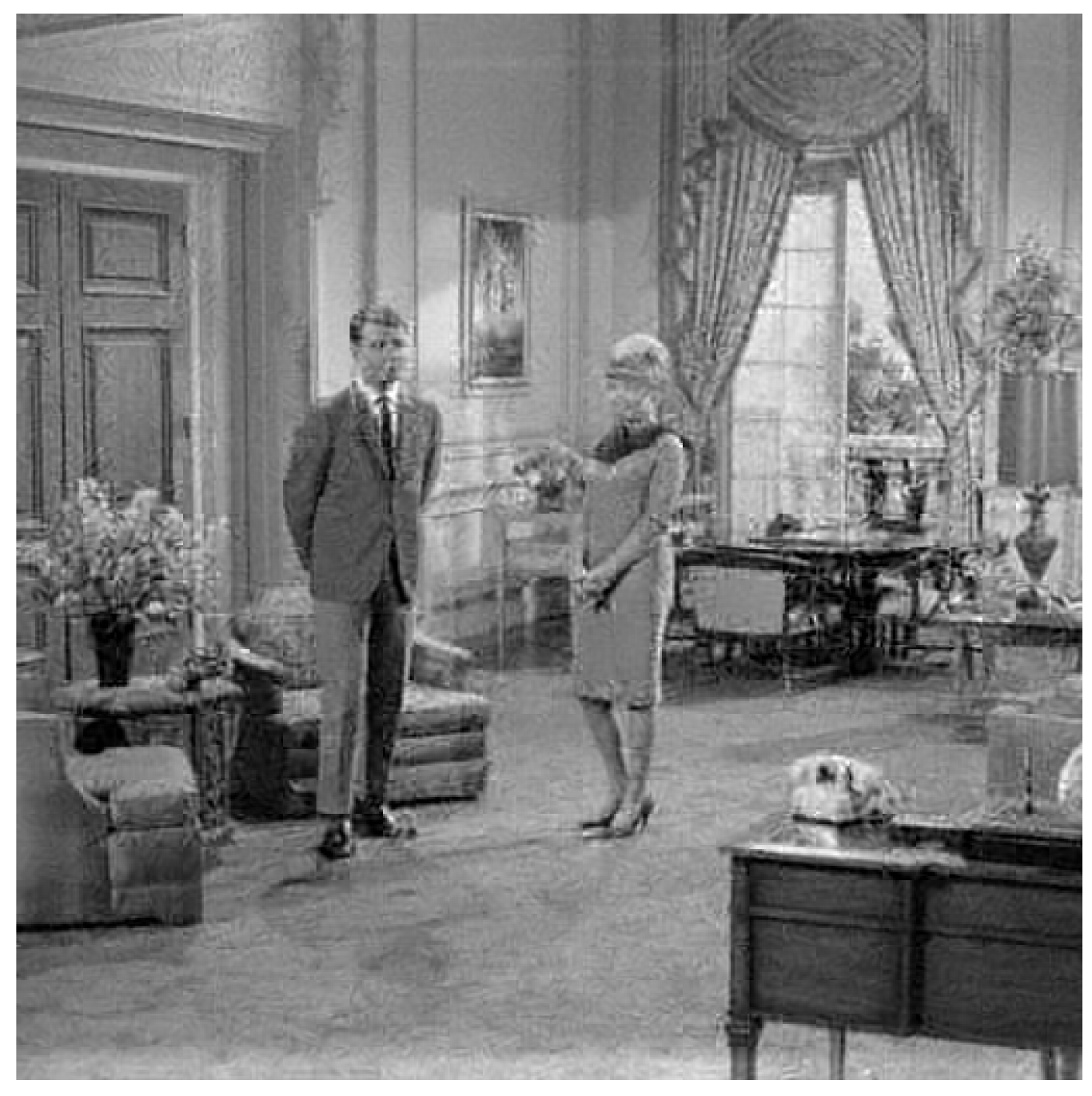}}
	\subfigure[CS$_{bispar}$: $27.33$]{\includegraphics[width=0.18\textwidth]{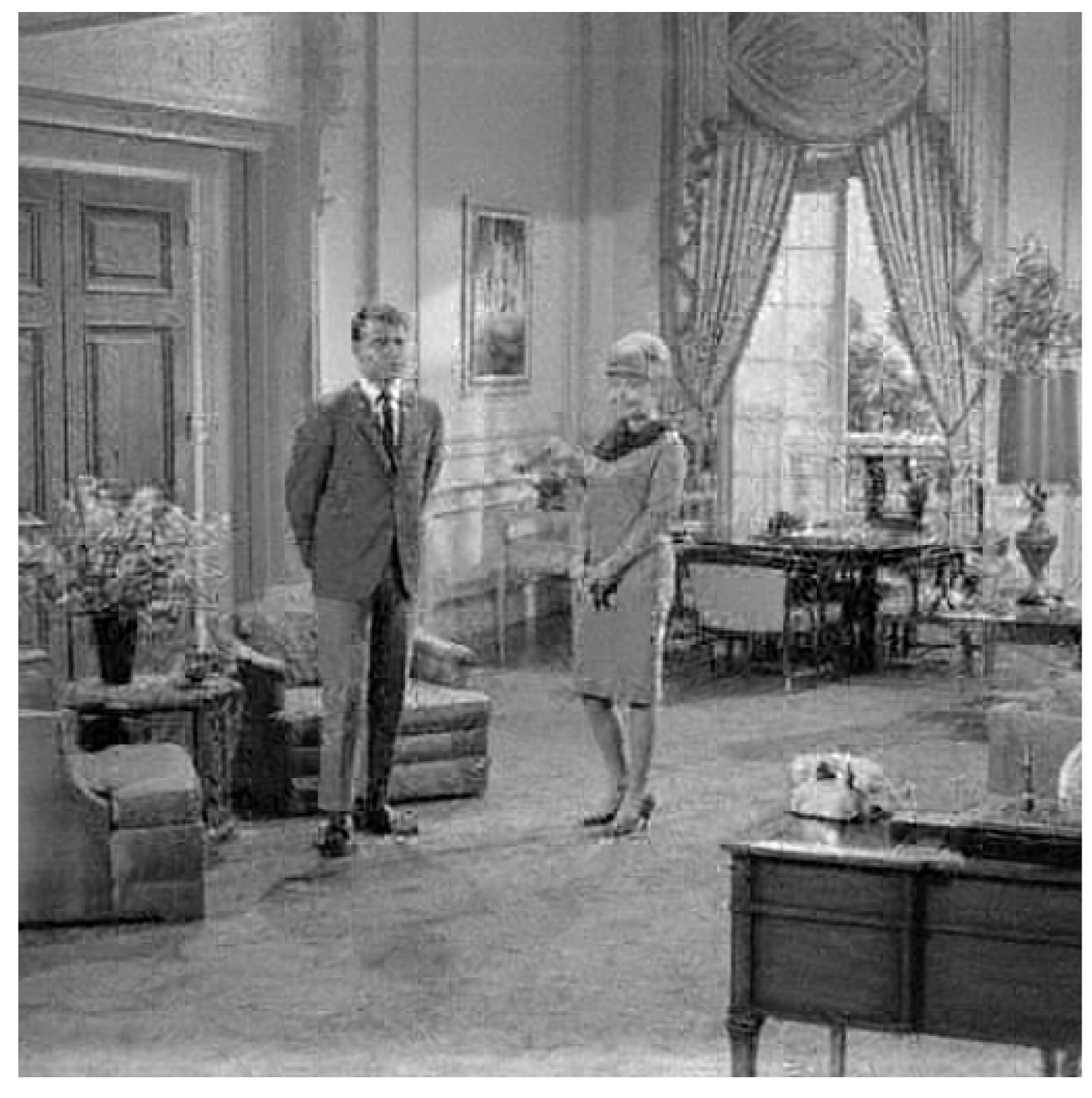}}
	\subfigure[CS$_{sparse}$: $29.69$]{\includegraphics[width=0.18\textwidth]{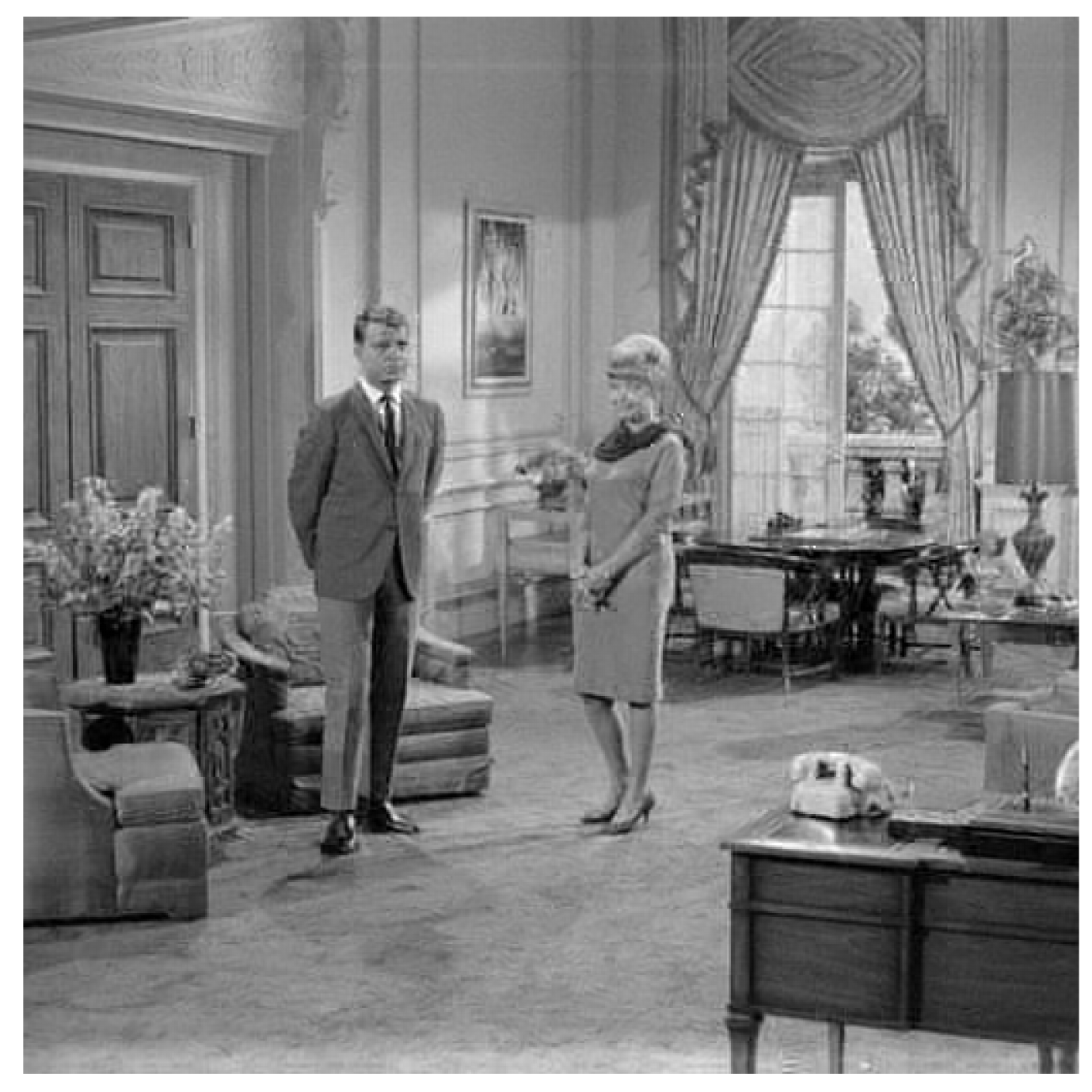}}
	\subfigure[{ CS$_{sparse-A}$: $30.47$}]{\includegraphics[width=0.18\textwidth]{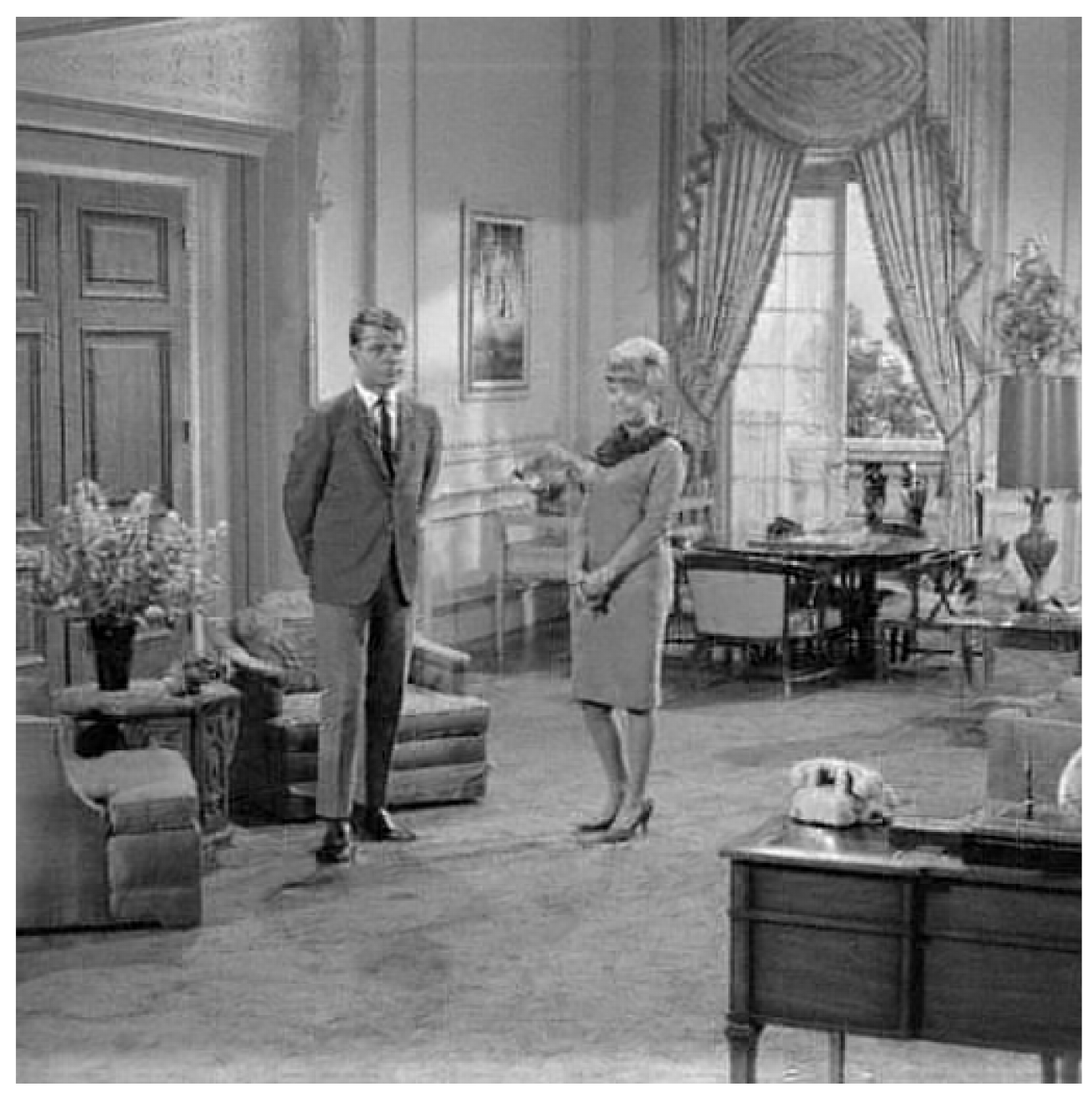}}
	\subfigure[CS$_{MT}$: $30.90$]{\includegraphics[width=0.18\textwidth]{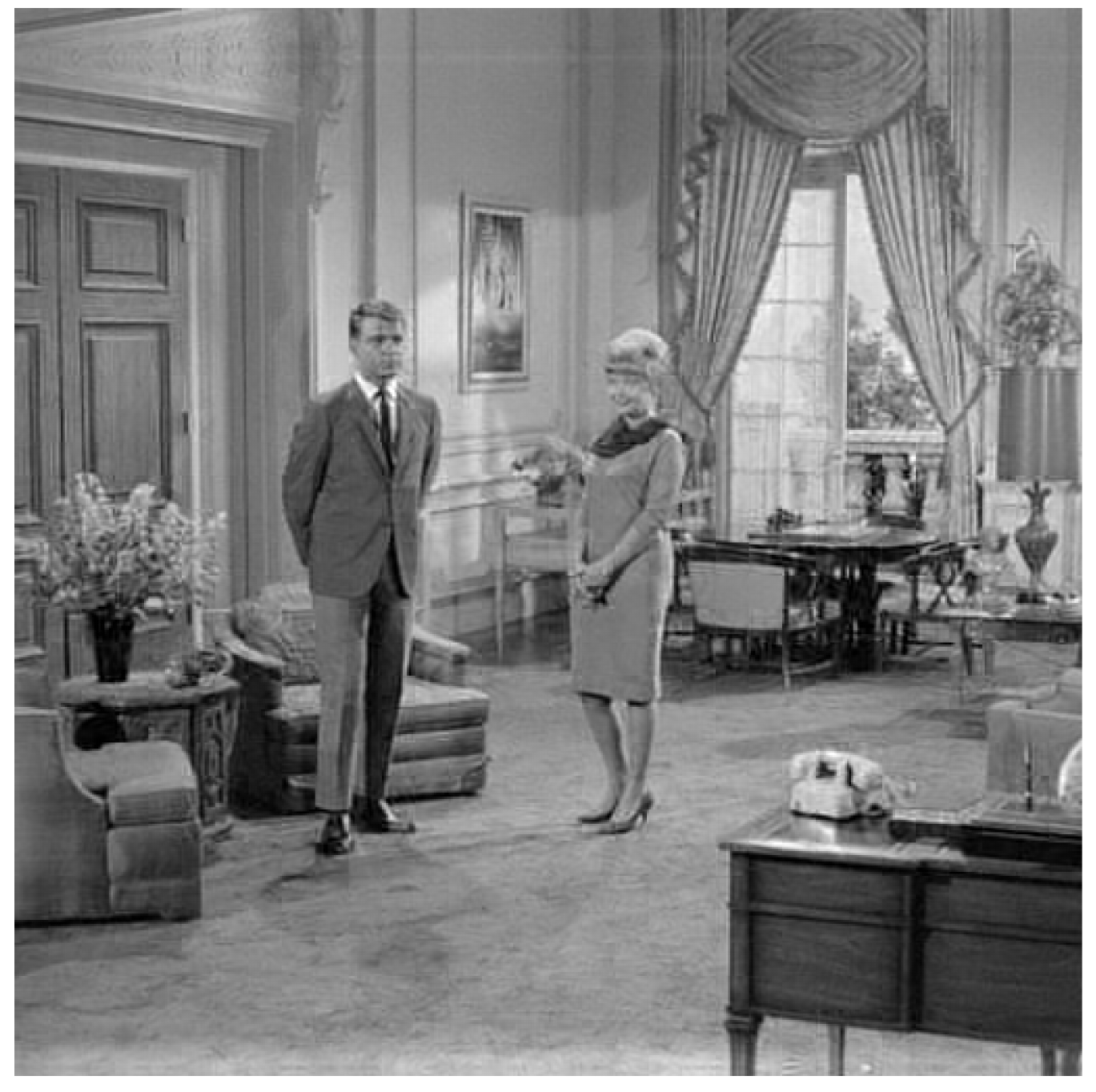}}

	\subfigure[$+0.51$]{\includegraphics[width=0.18\textwidth]{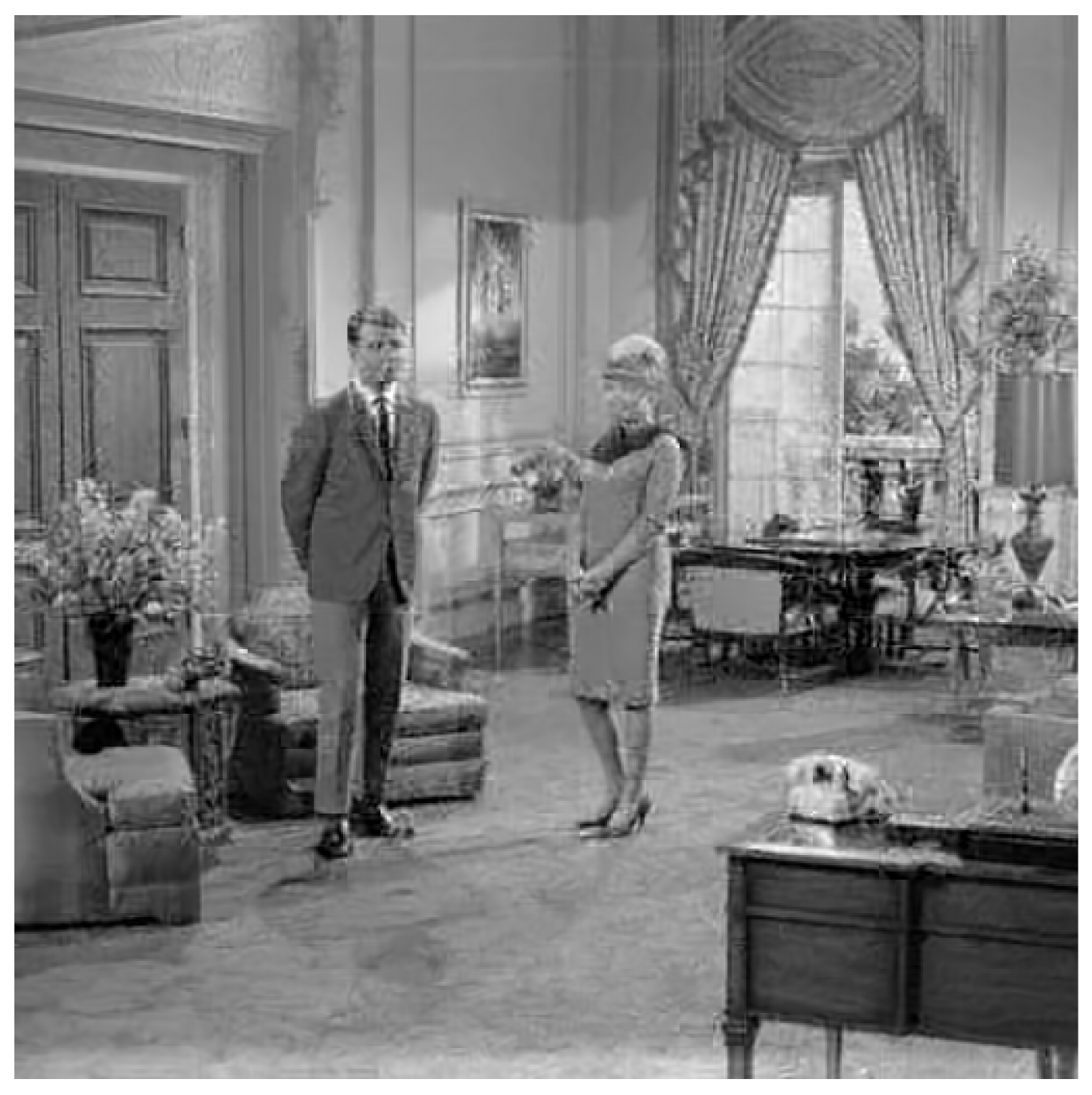}}    
	\subfigure[$+0.50$]{\includegraphics[width=0.18\textwidth]{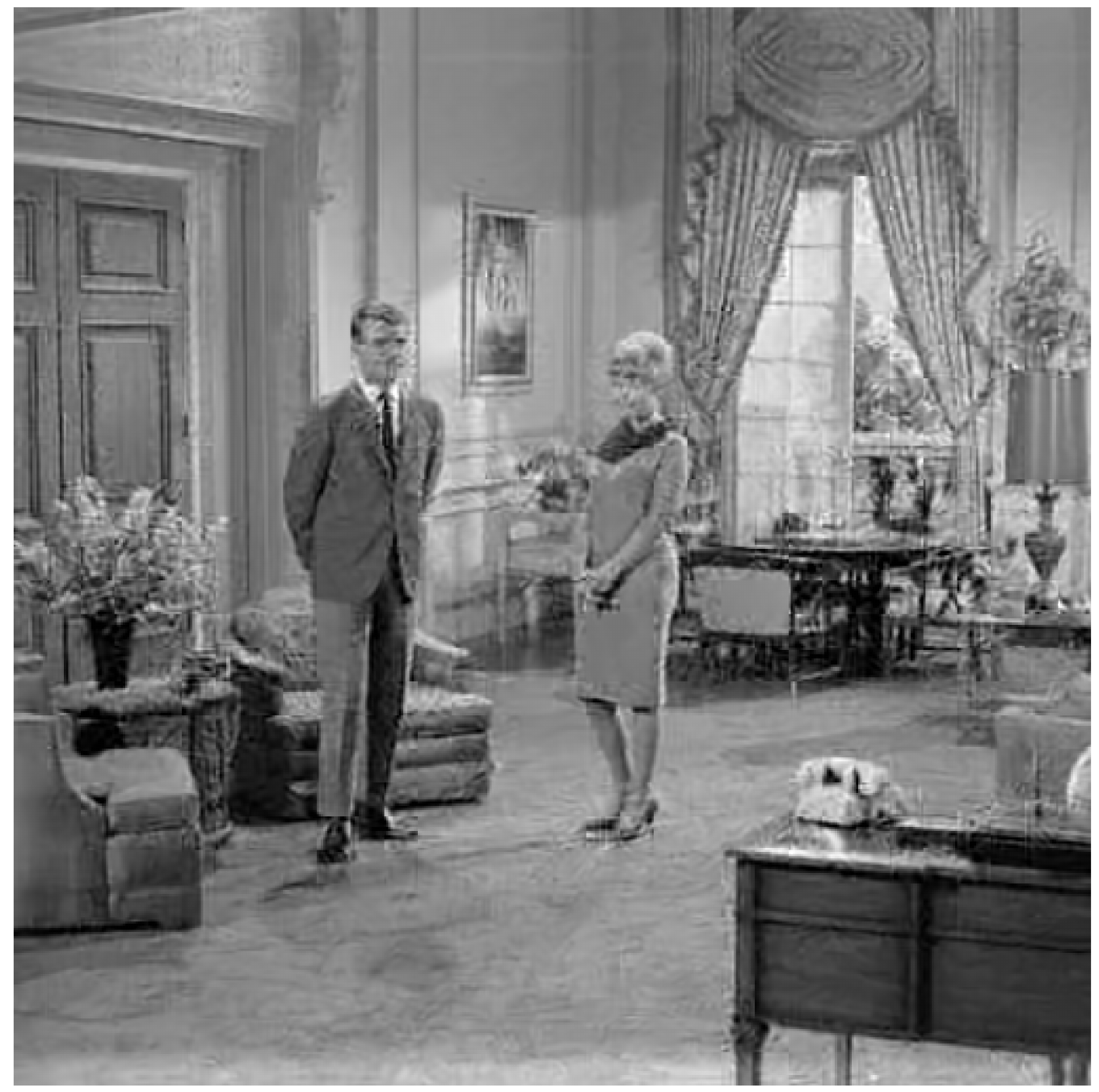}}    
	\subfigure[$+0.42$]{\includegraphics[width=0.18\textwidth]{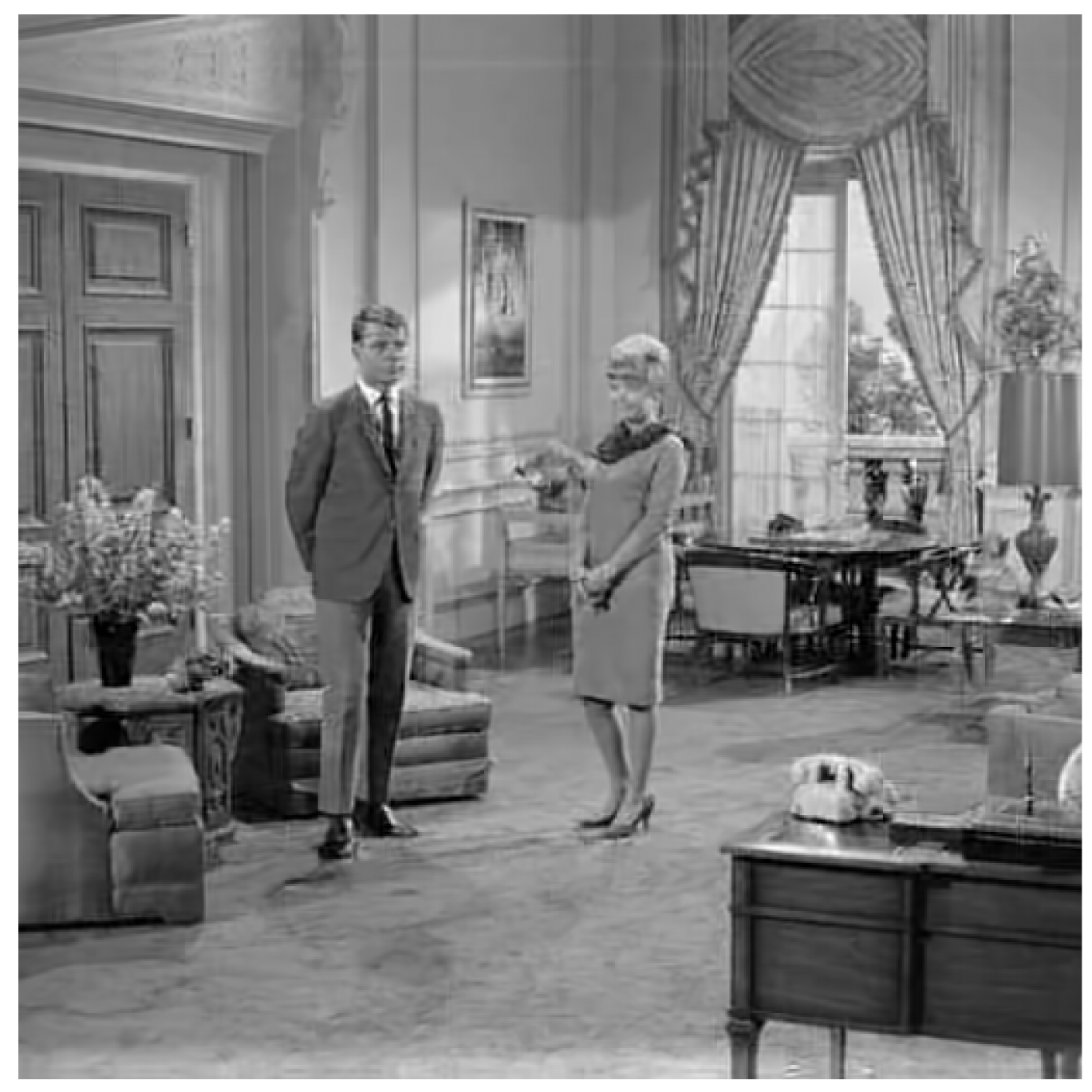}}
	\subfigure[$+0.37$]{\includegraphics[width=0.18\textwidth]{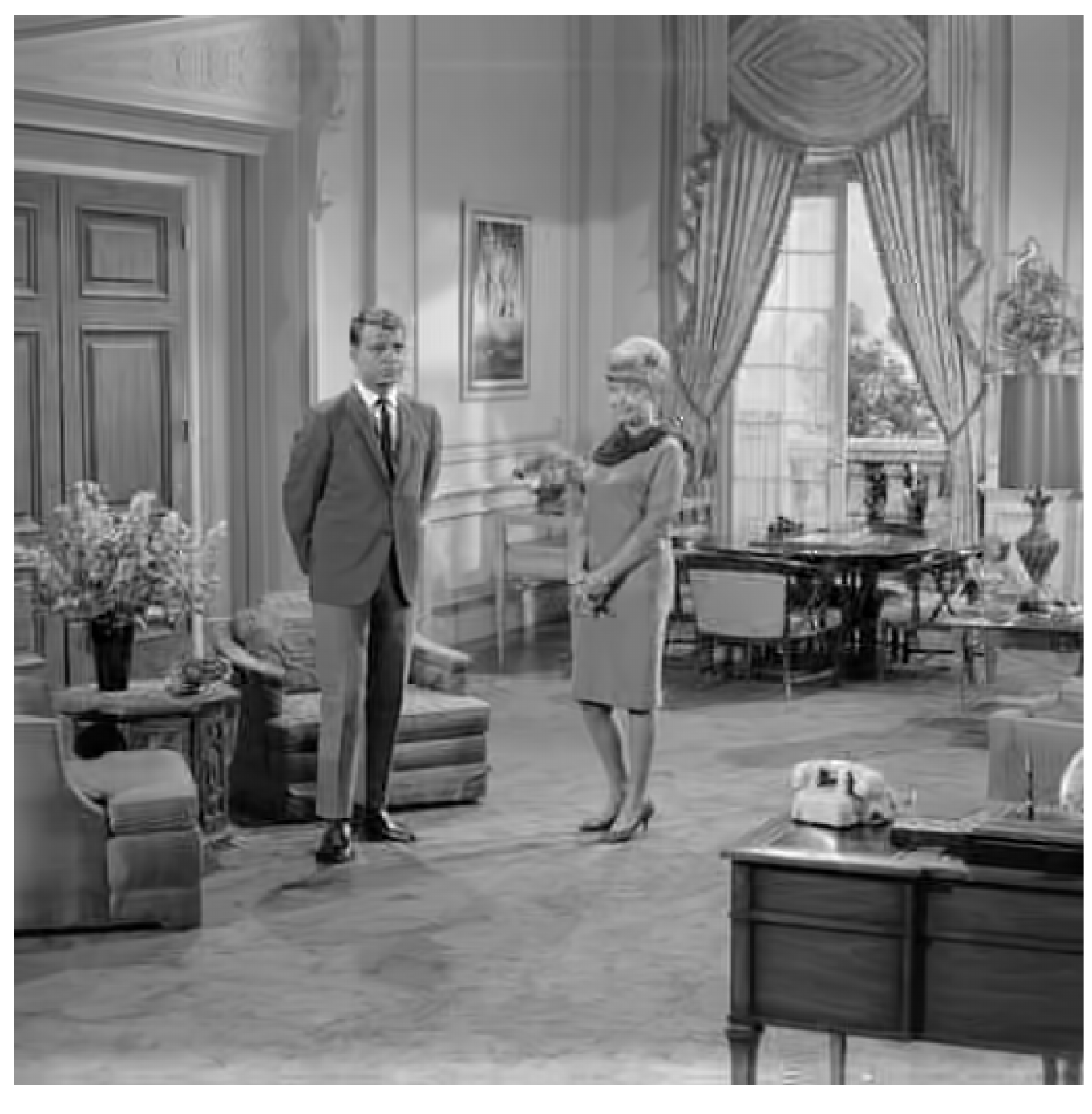}}
	\subfigure[$+0.23$]{\includegraphics[width=0.18\textwidth]{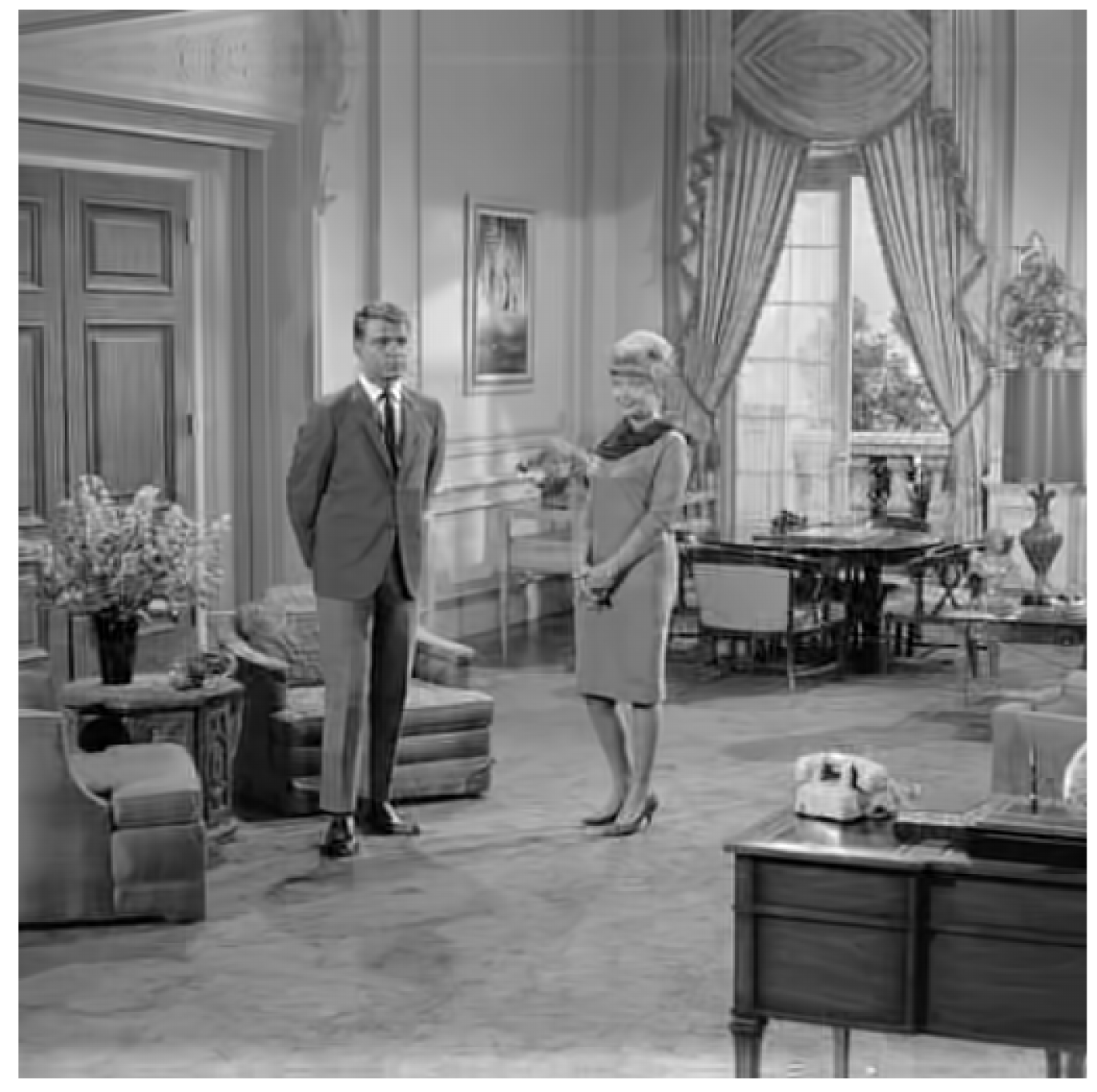}}
	
	\vspace{-0.2cm}
	\caption{The PSNR (dB) of reconstructed images for different CS systems with $M=80,~N=256,~L=800,~K = 16,~\lambda = 0.5,\kappa = 10$. Up: the PSNR without deblocking. Bottom: the amount of improved PSNR with deblocking.}
	\label{fig:images2} 
\end{figure*}

\section{Conclusions}
\label{sec:Conclusions}
We proposed a framework for designing a structured sparse sensing matrix that is robust to sparse representation error that widely exists for practical signals and can be efficiently implemented to sense signals. An alternating minimization-based algorithm is used to solve the optimal design problem, whose convergence is rigorously analyzed. The simulations demonstrate the promising performance of the proposed structured sparse sensing matrix in terms of signal reconstruction accuracy for synthetic data and real images.

As shown in Section~\ref{sec:simulations}, utilizing the base matrix $\mA$ can improve the performance of the obtained sparse sensing matrix, especially when the number of non-zeros in $\mPhi$ is very small. Thus, it is of interest to utilize a base matrix $\mA$ which has a few degrees of freedom (or parameters) that can be optimized, making it possible to simultaneously optimize the base matrix $\mA$ and the sparse matrix $\bm{\Phi}$. One choice of such an idea is to utilize a series of Givens rotation matrices, which have parameters for choosing and can be implemented very efficiently, to act as the base matrix. We refer this to future work. {It is also of interest to adopt the optimized structured sparse sensing matrix into analog-to-digital conversion system based on compressive sensing \cite{tropp2010beyond,davenport2012compressive}. Towards that end, it is important to develop a quantized (even 1-bit) sparse sensing matrix. 
We finally note that it remains an open problem to certify certain properties (such as the RIP) for the optimized sensing matrices \cite{elad2007optimized,duarte2009learning,abolghasemi2012gradient,li2013projection,chen2013projection,li2015designing,bai2015alternating,hong2016efficient,hong2017SP,li2017gradient,zhu2018collaborative}, which empirically outperforms a random one that satisfies the RIP. Works in these directions are ongoing.}

\section*{Acknowledgment}
This research is supported in part by ERC Grant agreement No. 320649, and in part by the Intel Collaborative Research Institute for Computational Intelligence (ICRI-CI).

\appendix
\section{Proof of \cref{lem:descent lemma}}
\label{prf:descent lemma}
\begin{proof}
	We parametrize the function value through the line passing $[\mPhi', \mPhi]$ by $t$, i.e.,  define $\upsilon(t) = f(\mPhi'+t(\mPhi-\mPhi'),\mG)$. It is clear that $\upsilon(0) = f(\mPhi',\mG), ~ \upsilon(1)=f(\mPhi,\mG)$. Then we have
	\begin{align*}
	&\upsilon(1) - \upsilon(0) = f(\mPhi,\mG) - f(\mPhi',\mG) = \int_{0}^{1}\upsilon'(t)dt \\
	& = \int_{0}^{1} \langle \nabla_{\mPhi} f(\mPhi'+t(\mPhi-\mPhi'),\mG), \mPhi-\mPhi'  \rangle   dt\\
	&= \int_{0}^{1} \langle \nabla_{\mPhi} f(\mPhi'+t(\mPhi-\mPhi'),\mG) - \nabla_{\mPhi} f(\mPhi',\mG), \mPhi-\mPhi'  \rangle   dt\\
	& \quad + \langle\nabla_{\mPhi} f(\mPhi',\mG), \mPhi-\mPhi'  \rangle \\
	&\leq  \int_{0}^{1} \| \nabla_{\mPhi} f(\mPhi'+t(\mPhi-\mPhi'),\mG)- \nabla_{\mPhi} f(\mPhi',\mG)\|_Fdt \\
	&\quad \cdot \|\mPhi-\mPhi' \|_F + \langle\nabla_{\mPhi} f(\mPhi',\mG), \mPhi-\mPhi'  \rangle \\
	&\leq L\|\mPhi-\mPhi' \|_F^2 \int_{0}^{1} tdt+ \langle\nabla_{\mPhi} f(\mPhi',\mG), \mPhi-\mPhi'  \rangle \\
	&=\frac{L}{2}\|\mPhi-\mPhi' \|_F^2 + \langle\nabla_{\mPhi} f(\mPhi',\mG), \mPhi-\mPhi'  \rangle,
	\end{align*}
	where in the last inequality we have used \eqref{eq:Lipschitz gradient}.
\end{proof}

\section{Proof of \cref{thm:subsequence convergence}}
\label{sec:prf thm subsequence convergence}

We first state the following definition of subdifferential for a general lower semicontinuous function, which is not necessarily differentiable.
\begin{defi}(Subdifferentials~\cite{attouch2010proximal}) Let $\sigma:\R^d\rightarrow (-\infty,\infty]$ be a proper and lower semicontinuous function, whose domain is defined as
\[
\domain \sigma:=\left\{\vu\in\R^d:\sigma(\vu)<\infty\right\}.
\]

The Fr\'{e}chet subdifferential $\partial_F \sigma$ of $\sigma$ at $\vu$ is defined by
    \[
    \partial_F \sigma(\vu) = \left\{\vz:\lim_{\vv\rightarrow \vu}\inf\frac{\sigma(\vv) - \sigma(\vu) - \langle \vz, \vv - \vu\rangle}{\|\vu - \vv\|}\geq 0\right\}
    \]
for any $\vu\in \domain \sigma$ and $\partial_F \sigma(\vu) = \emptyset$ if $\vu\notin \domain \sigma$.

The subdifferential $\partial \sigma(\vu)$ of $\sigma$ at $\vu\in \domain \sigma$ is defined as follows
\[
\partial \sigma = \left\{
\vz:\exists \vu_k\rightarrow \vu, \sigma(\vu_k)\rightarrow \sigma(\vu), \vz_k\in \partial_F \sigma(\vu_k)\rightarrow \vz
\right\}
\]
\end{defi}

We say $\vu$ a critical point (a.k.a. stationary point) if subdifferential at $\vu$ is $\vzero$. The set of critical points of $\sigma$ is denoted by $\setC(\sigma)$.

\begin{proof}[Proof of Theorem~\ref{thm:subsequence convergence}]
We prove Theorem~\ref{thm:subsequence convergence} by individually proving the four arguments.

Show (P1): It is clear that for any $k\in\N$, $\mPhi_k\in\setS_\kappa$ and $\mG_k\in\setG_\xi$. Thus we have $\rho(\mPhi_k,\mG_{\ell}) = f(\mPhi_k,\mG_{\ell})$ for any $k,\ell\in\N$. Let $\mB = \mPsi^\T\mPhi_{k+1}^\T\mPhi_{k+1}\mPsi$. Noting that $\setG_\xi$ is a closed  convex set, we have
\[
\left\langle\mG - \calP_{\setG_\xi}(\mB),\calP_{\setG_\xi}(\mB) - \mB  \right\rangle\geq 0, \ \forall \mG \in \setG_\xi,
\]
which directly implies
\begin{align*}
&\rho(\mPhi_{k+1},\mG_k) - \rho(\mPhi_{k+1},\mG_{k+1}) \\&= \|\mG_k - \mB\|^2_F - \|\mG_{k+1} - \mB\|^2_F\\
& = \|\mG_k - \calP_{\setG_\xi}(\mB) + \calP_{\setG_\xi}(\mB) -\mB\|^2_F - \|\calP_{\setG_\xi}(\mB) - \mB\|^2_F\\
& = \|\mG_k - \calP_{\setG_\xi}(\mB)\|_F^2 + 2 \left\langle\mG_k - \calP_{\setG_\xi}(\mB), \calP_{\setG_\xi}(\mB) -\mB\right\rangle\\
& \geq  \|\mG_k - \mG_{k+1}\|_F^2\geq 0.
\end{align*}
On the other hand, we rewrite \eqref{eq:update Phi} as
\begin{equation}\begin{split}
\mPhi_{k+1} &\in \calP_{\calS_\kappa}(\mPhi_{k} - \eta\nabla f(\mPhi_{k},\mG_{k}))\\
& = \argmin_{\mZ\in\setS_\kappa}\|\mZ -\left(\mPhi_{k} - \eta\nabla f(\mPhi_{k},\mG_{k})\right)\|_F^2\\
&\in\argmin_{\mZ\in\setS_\kappa} h_{1/\eta}(\mZ,\mPhi_k,\mG_{k})
\end{split}\label{eq:rewrite Phi}\end{equation}
which implies that
\begin{align*}
h_{1/\eta}(\mPhi_{k+1},\mPhi_k,\mG_{k})&\leq h_{1/\eta}(\mPhi_{k},\mPhi_{k},\mG_{k})\\& = f(\mPhi_k,\mG_{k}).
\end{align*}
This along with Lemma~\ref{lem:descent lemma} gives
\begin{align*}
&f(\mPhi_k,\mG_{k}) - f(\mPhi_{k+1},\mG_{k})\\ & \geq f(\mPhi_k,\mG_{k}) - h_{L_c}(\mPhi_{{k+1}},\mPhi_{k},\mG_{k})\\
& \geq h_{1/\eta}(\mPhi_{k+1},\mPhi_k,\mG_{k}) - h_{L_c}(\mPhi_{{k+1}},\mPhi_{k},\mG_{k})\\
& = \frac{\frac{1}{\eta} - L_c}{2}\| \mPhi_k - \mPhi_{k+1}\|_F^2.
\end{align*}

Show (P2): It follows from \eqref{eq:sufficient decrease} that
\begin{align*}
&\rho_0\geq \rho(\mPhi_1,\mG_0) \geq \rho(\mPhi_1,\mG_{1})\geq \cdots\rho(\mPhi_k,\mG_{k})\\& \geq \rho(\mPhi_{k+1},\mG_{k}) \geq \rho(\mPhi_{k+1},\mG_{k+1})\geq \cdots
\end{align*}
which together with the fact that $\rho(\mPhi_k,\mPsi_k)\geq 0$ gives the convergence of sequence $\{\rho(\mPhi_k,\mG_k)\}_{k\geq 0}$. This also implies that $(\mPhi_k,\mG_k)\in \setL_{\rho_0}$ and hence $\{(\mPhi_k,\mG_k)\}_{k\geq 0}$ is a bounded sequence.

Show (P3): Utilizing \eqref{eq:sufficient decrease} for all $k\in\N$ and summing them together, we obtain
\begin{align*}
&\sum_{k=0}^\infty \frac{\frac{1}{\eta} - L_c}{2}\left(\|\mPhi^k - \mPhi^{k+1}\|_F^2\right) + \|\mG^k - \mG^{k+1}\|_F^2\\
&\leq \rho_0 - \lim_{k\rightarrow \infty} \rho(\mPhi_k,\mG_{k})\leq \rho_0,
\end{align*}
which implies that the series $\{\sum_{k=0}^n \|\mPhi^k - \mPhi^{k+1}\|_F^2 + \|\mG^k - \mG^{k+1}\|_F^2\}_n$ is convergent. This together with the fact that $ \|\mPhi^k - \mPhi^{k+1}\|_F^2\geq 0$ and $\|\mG^k - \mG^{k+1}\|_F^2\geq 0$ gives \eqref{eq:diff goes to 0}.

Show (P4): We rewrite \eqref{eq:rewrite Phi} as
\e
\mPhi_{k+1} \in \argmin_{\mPhi\in \R^{M\times N}} h_{\frac{1}{\eta}}(\mPhi,\mPhi_{k},\mG_{k}) + \delta_{\setS_\kappa}(\mPhi),
\label{eq:rewrite Phi 2}\ee
which implies (by the optimality of $\mPhi_{k+1}$ in \eqref{eq:rewrite Phi 2} and letting $\mPhi = \underline{\mPhi}$, i.e. the limit  of a convergent subsequence $\{\mPhi_{k'}\}_{k'}$)
\begin{align*}
&\langle \nabla_{\mPhi} f(\mPhi_k,\mG_{k}), \mPhi_{k+1}-\mPhi_{k}  \rangle + \frac{\eta}{2}\|\mPhi_{k+1} - \mPhi_k\|_F^2 + \delta_{\setS_\kappa}(\mPhi_{k+1})\\
&\leq\langle \nabla_{\mPhi} f(\mPhi_k,\mG_{k}), \underline{\mPhi}-\mPhi_{k}  \rangle + \frac{\eta}{2}\|\underline{\mPhi} - \mPhi_{k}\|_F^2 + \delta_{\setS_\kappa}(\underline{\mPhi}).
\end{align*}
This further gives (take limit on subsequence $\{\mPhi_{k'}\}_{k'}$)
\begin{equation}\begin{split}
&\limsup_{k'\rightarrow \infty}\delta_{\setS_\kappa}(\mPhi_{k'}) - \delta_{\setS_\kappa}(\underline{\mPhi}) \\
&\leq \limsup_{k'\rightarrow \infty} \langle \nabla_{\mPhi} f(\mPhi_{k'-1},\mG_{k'-1}), \underline{\mPhi}-\mPhi_{k'}  \rangle\\
 &\quad + \frac{\eta}{2}\|\mPhi_{k'-1} - \underline{\mPhi}\|_F^2- \frac{\eta}{2}\|\mPhi_{k'} - \mPhi_{k'-1}\|_F^2\\
& = 0,
\end{split}\label{eq:lim sup}\end{equation}
where the last line follows from \eqref{eq:diff goes to 0}, the fact that scalar product is continuous and $\lim_{k'\rightarrow\infty}\|\mPhi_{k'-1} - \underline{\mPhi}\|_F = 0$ since
\begin{align*}
0&\leq \lim_{k'\rightarrow\infty}\|\mPhi_{k'-1} - \underline{\mPhi}\|_F\\& = \lim_{k'\rightarrow\infty}\|\mPhi_{k'-1} -\mPhi_{k'} +\mPhi_{k'}  - \underline{\mPhi}\|_F\\& \leq  \lim_{k'\rightarrow\infty}\|\mPhi_{k'} - \underline{\mPhi}\|_F + \|\mPhi_{k'} - \mPhi_{k'-1}\|_F = 0.
\end{align*}
From the fact that $\delta_{\setS_\kappa}(\mPhi)$ is lower semi-continuous, we have
 \[
 \delta_{\setS_\kappa}( \underline{\mPhi})  \leq \liminf_{k'\rightarrow \infty}\delta_{\setS_\kappa}(\mPhi_{k'}).\]
Utilizing \eqref{eq:lim sup}  gives
\[
\limsup_{k'\rightarrow \infty}\delta_{\setS_\kappa}(\mPhi_{k'})  \leq \delta_{\setS_\kappa}(\underline{\mPhi})\leq \liminf_{k'\rightarrow \infty}\delta_{\setS_\kappa}(\mPhi_{k'}),
\]
which together with the fact 
\[\liminf_{k'\rightarrow \infty}\delta_{\setS_\kappa}(\mPhi_{k'})  \leq \limsup_{k'\rightarrow \infty}\delta_{\setS_\kappa}(\mPhi_{k'})
\] gives
 \[
 \delta_{\setS_\kappa}(\underline{\mPhi}) = \liminf_{k'\rightarrow \infty}\delta_{\setS_\kappa}(\mPhi_{k'})    = \limsup_{k'\rightarrow \infty}\delta_{\setS_\kappa}(\mPhi_{k'}),
 \]
and hence
 \[
\lim_{k'\rightarrow \infty}\delta_{\setS_\kappa}(\mPhi_{k'})  = \delta_{\setS_\kappa}(\underline{\mPhi}).
\]

Since $\setG_\xi$ is a compact set and $\mG_{k'}\in \setG_\xi, \ \forall \ k'\in\N$, we have the limit point $\underline{\mG} \in \setG_\xi$. Therefore, we obtain
\begin{align*}
 \lim_{k'\rightarrow \infty}\rho(\mPhi_{k'},\mG_{k'})& =  \lim_{k'\rightarrow \infty} f(\mPhi_{k'},\mG_{k'})+ \delta_{\setS_\kappa}(\mPhi_{k'}) + \delta_{\setG_\xi}(\mG_{k'})\\& =  \rho(\underline{\mPhi},\underline{\mG}).
\end{align*}

The remaining part is to prove that $\underline{\mW} = (\underline{\mPhi},\underline{\mG})$ is a stationary point of $\rho$, which is equivalent to show $(\vzero,\vzero)\in \partial \rho(\underline{\mPhi},\underline{\mG})$. In what follows, we show a stronger result that $(\mzero,\mzero)\in\lim_{k\rightarrow \infty} \partial \rho(\mPhi_k,\mG_k)$.

First note that
\begin{align*}
\mG_{k} &= \argmin_{\mG\in\R^{L\times L}} \rho(\mPhi_{k},\mG)
\end{align*}
The optimality condition gives \cite{bolte2014proximal}
\e
\mzero =  \underbrace{\nabla_{\mG} f(\mPhi_k,\mG_k) + \mU_k}_{\mD_{\mG_k}}\in \partial_{\mG} \rho(\mW_k),
\label{eq:partial G}\ee
where $\mU_k \in \partial \delta_{\setG_\xi}(\mG_k)$. On the other hand,
the optimality condition of \eqref{eq:rewrite Phi 2} gives (by setting $k \leftarrow k-1$ in \eqref{eq:rewrite Phi 2})
\[
\nabla_{\mPhi} f(\mPhi_{k-1},\mG_{k-1}) + \frac{1}{\eta}(\mPhi_k - \mPhi_{k-1}) + \mV_k = \vzero,
\]
where $\mV_k \in \partial \delta_{\setS_\kappa}(\mPhi_k)$. Thus we have
\[
\underbrace{\nabla_{\mPhi} f(\mW_{k}) -\nabla_{\mPhi} f(\mW_{k-1}) - \frac{1}{\eta}(\mPhi_k - \mPhi_{k-1})}_{\mD_{\mPhi_k}} \in \partial_{\mPhi} \rho(\mW_{k}),
\]
which along with \eqref{eq:partial G} gives
\begin{align}
&\left\|(\mD_{\mPhi_k},\mD_{\mG_k})\right\|_F = \left\|\mD_{\mPhi_k}\right\|_F \nonumber\\
&= \| \nabla_{\mPhi} f(\mW_{k}) -\nabla_{\mPhi} f(\mW_{k-1}) - \frac{1}{\eta}(\mPhi_k - \mPhi_{k-1}) \|_F \nonumber\\
&\leq \| \nabla_{\mPhi} f(\mW_{k}) -\nabla_{\mPhi} f(\mPhi_k,\mG_{k-1})\| \nonumber\\
& \quad + \| \nabla_{\mPhi} f(\mPhi_k,\mG_{k-1}) -\nabla_{\mPhi} f(\mW_{k-1})\| \nonumber\\
&\quad + \frac{1}{\eta}\|\mPhi_k - \mPhi_{k-1}\|_F \nonumber\\
& \leq L_c \|\mG_k - \mG_{k-1}\|_F + (L_c + \frac{1}{\eta}) \|\mPhi_k - \mPhi_{k-1}\|_F \nonumber\\
&   \leq (2L_c + \frac{1}{\eta}) \|\mW_k - \mW_{k-1}\|_F,       \label{safeguard}
\end{align}
where we have used the Lipschitz gradient  \eqref{eq:Lipschitz gradient} in the second inequality.

Applying~\eqref{eq:diff goes to 0}, we finally obtain
\[
\lim_{k\rightarrow \infty} (\mD_{\mPhi_{k}},\mD_{\mG_{k}}) = (\mzero,\mzero)
\]
since
\[
\lim_{k\rightarrow \infty} \left\|(\mD_{\mPhi_{k}},\mD_{\mG_{k}})\right\|_F = 0.
\]
Thus $(\mzero,\mzero)\in\lim_{k\rightarrow \infty} \partial \rho(\mPhi_k,\mG_k)$ and we conclude that  any convergent subsequence of $\{\mW_k\}$ converges to a stationary point of \eqref{eq:unconstrainded problem}.

Finally, the statement
\[
\lim_{k\rightarrow \infty}\rho(\mW_k) = \rho(\underline{\mW}).
\]
directly follows from (P2) that the objective value sequence $\{\rho(\mW_k)\}_{k\in\N}$ is convergent.

\end{proof}

\section{Proof of \cref{thm:sequence convergence}}
\label{sec:prf thm sequence convergence}
We first state the definition of Kurdyka-Lojasiewicz (KL) inequality, which is proved to be useful for convergence analysis~\cite{attouch2010proximal,attouch2013convergence,bolte2014proximal}.
\begin{defi}\label{def:KL}
	A proper semi-continuous  function $\sigma(\vu)$ is said  to  satisfy Kurdyka-Lojasiewicz (KL) inequality,  if  $\underline{\vu}$ is a stationary point of $\sigma(\vu)$, then $\exists ~\delta>0,~\theta\in[0,1),~C_1>0,~s.t.$
\[
\left|\sigma(\vu) - \sigma(\underline{\vu})\right|^{\theta}  \leq C_1 \|\vv\|,~~\forall~\vu\in B(\underline{\vu}, \delta),~\forall~\vv\in \partial \sigma(\vu)
\]
\end{defi}

It is clear that our objective function $\rho(\mPhi, \mG)$ is lower semi-continuous and it satisfies the above KL inequality since the three components $f(\mPhi,\mG)$, $\delta_{\setS_\kappa}(\mPhi)$  and $\delta_{\setG_\xi}(\mG)$ all have the KL inequality~\cite{attouch2013convergence,bolte2014proximal}.

\begin{proof}[Proof of Theorem~\ref{thm:sequence convergence}]
\Cref{thm:subsequence convergence} reveals the subsequential convergence property of the iterates sequence $\{\mW_k= (\mPhi_{k},\mG_{k})\}_k$, i.e., the limit point of any convergent subsequence converges to a stationary point. In what follows, we show the sequence $\{\mW_k= (\mPhi_{k},\mG_{k})\}_k$ itself is indeed convergent, and hence it converges to a certain stationary point of $\underline{\mW} = (\underline{\mPhi},\underline{\mG})$.

It follows from \eqref{eq:lim f = f lim} that for any $\delta >0$, there exists an integer $n$ such that $\mW_k\in B(\underline{\mW}, \delta),~\forall~k>n$ for some  stationary point $\underline{\mW} \in \calC(\rho)$. From the concavity of  the function $h(y)= y^{1-\theta}$ with domain $ y>0$, we have\footnote{If a differential function $f(\vx)$ is concavity, the following inequality holds: $f(\vy)-f(\vx)\leq \langle\nabla f(\vx),\vy-\vx\rangle$.}
\begin{align}\label{convexity}
&\left[  \rho(\mW_{k+1}) - \rho(\underline{\mW})    \right]^{1-\theta}\\
 &\leq \left[  \rho(\mW_{k}) - \rho(\underline{\mW})   \right]^{1-\theta} + (1-\theta) \frac{\rho(\mW_{k+1})-\rho(\mW_{k})}{ \left[  \rho(\mW_{k}) - \rho(\underline{\mW})    \right]^{\theta}}.\nonumber
\end{align}
We now provide lower bound and upper bound for $\rho(\mW_{k}) - \rho(\mW_{k+1})$ and $\left[  \rho(\mW_{k}) - \rho(\underline{\mW})    \right]^{\theta}$, respectively. It follows from \eqref{eq:sufficient decrease} that
\[
\rho(\mW_{k}) - \rho(\mW_{k+1}) \geq C_2\|\mW_{k+1} - \mW_{k}\|_F^2, \]
where $C_2 = \min\{\frac{\frac{1}{\eta} - L_c}{2}, 1\}$.  On the other hand, from \eqref{safeguard} and  the KL inequality we have
\begin{align*}
\left[  \rho(\mW_{k}) - \rho(\underline{\mW})    \right]^{\theta}& \leq C_1 \left\|(\mD_{\mPhi_k},\mD_{\mG_k})\right\|_F\\&\leq C_3\|\mW_{k} - \mW_{k-1}\|_F,
\end{align*}
where $C_3 = C_1(2L_c+1/\eta)$. Plugging the above two inequalities into \eqref{convexity} gives
\begin{align}\label{convexity2}
&\left[  \rho(\mW_{k}) - \rho(\underline{\mW})   \right]^{1-\theta} - \left[  \rho(\mW_{k+1}) - \rho(\underline{\mW})    \right]^{1-\theta} \nonumber\\
&\geq  (1-\theta) \frac{C_2\|\mW_{k+1} - \mW_{k}\|_F^2 }{ C_3\|\mW_{k} - \mW_{k-1}\|_F}.
\end{align}
Let $C_4 = (1-\theta)C_2/C_3$. Repeating the above equation for $k$ from $1$ to $\infty$ and summing them gives
\begin{align*}\label{convexity3}
&\frac{1}{C_4}\left[  \rho(\mW_{1}) - \rho(\underline{\mW})   \right]^{1-\theta} - \frac{1}{C_4}\left[  \rho(\mW_{\infty}) - \rho(\underline{\mW})    \right]^{1-\theta} \\
&\geq  \sum_{k = 1}^{\infty}\frac{\|\mW_{k+1} - \mW_{k}\|_F^2 }{ \|\mW_{k} - \mW_{k-1}\|_F} + \|\mW_{k} - \mW_{k-1}\|_F\\
& \quad - \|\mW_{k} - \mW_{k-1}\|_F       \\
&\stackrel{(i)}{\geq}  2\sum_{k = 1}^{\infty}\|\mW_{k+1} - \mW_{k}\|_F - \sum_{k = 1}^{\infty}\|\mW_{k} - \mW_{k-1}\|_F       \\
& = \sum_{k = 1}^{\infty}\|\mW_{k+1} - \mW_{k}\|_F   - \|\mW_{1} - \mW_{0}\|_F,
\end{align*}
where $(i)$ is from the arithmetic inequality, i.e., $a^2+b^2\geq 2ab$. The proof is completed by applying the above result with the boundedness of $\{\mW_k\}_k$ and \eqref{eq:lim f = f lim}:
 $$\sum_{k = 1}^{\infty}\|\mW_{k+1} - \mW_{k}\|_F< \infty$$
which implies that the sequence $\{\mW_{k}\}_{k\in\N}$ is Cauchy \cite{bolte2014proximal} in a compact set and hence it is convergent.
%which implies the series $\{\sum_{k = 1}^{m}\|\mW_{k+1} - \mW_{k}\|_F\}_m$ is convergent. Assume $m_2>m_1$, hence we have
%$$
%\limsup_{m\rightarrow \infty, m_1,m_2\geq m }   \sum_{k = 1}^{m_2}\|\mW_{k+1} - \mW_{k}\|_F  -  \sum_{k = 1}^{m_1}\|\mW_{k+1} - \mW_{k}\|_F =0,
%$$
%and hence
%\begin{align*}
%0&\geq \limsup_{m\rightarrow \infty, m_1,m_2\geq m } \left\|\sum_{k = m_1+1}^{m_2}\mW_{k+1} - \mW_{k}\right\|_F \\
%&=\limsup_{m\rightarrow \infty, m_1,m_2\geq m } \left\|\mW_{m_2+1} - \mW_{m_1+1}\right\|_F \\
%&\geq 0.
%\end{align*}
%Therefore, the sequence $\{\mW_{k}\}_{k\in\N}$ is Cauchy in a compact set and hence is convergent.
\end{proof}

%\begin{comment}
\section{The Choice of Step Size With Unknown Lipschitz Constant $L_c$}
\label{Backtracking:unknow LC}
In practice, it is challenge to choose an appropriate step size since it is not easy to compute the Lipschitz constant $L_c$. According to the given convergence analysis in Section \ref{sec:proposed algorithm}, we know if the step size is chosen  to satisfy \eqref{eq:sufficient decrease}, the convergence is still guaranteed. Thus, we can utilize the backtracking method \cite{nocedal2006numerical} with inequality \eqref{eq:sufficient decrease} to search an appropriate step size without knowing $L_c$. The procedure is detailed in Algorithm~\ref{alg:  Backtracking}.

\begin{algorithm}[!htb]
	\caption{Backtracking Procedure}% with considering projection noise}
	\label{alg:  Backtracking}
	\begin{algorithmic}[1]
		\REQUIRE ~\\
		 Initial value: $(\eta_0,\gamma,\alpha)$ where $\eta_0$ is the initial guess step size, $\gamma\in(0,1)$ and $\alpha\in(0,1)$.
		\lastcon ~\\          %OUTPUT
		Step size $\eta$ and Updated $\mPhi_{k+1}$.
		%\ENSURE
		\STATE $ \eta \leftarrow \eta_0 $
		\STATE $\mPhi_{k+1}  \leftarrow \calP_{\calS_{\kappa}} (\mPhi_k-\eta \nabla_{\mPhi}f(\mPhi_k,\mG_{k}))$
		\WHILE {$\rho(\mPhi_k,\mG_{k}) - \rho(\mPhi_{k+1},\mG_{k})<\frac{\gamma}{2\eta}\|\mPhi_{k+1} - \mPhi_{k}\|_F^2$}
		\STATE $\eta \leftarrow \alpha \eta$
		\STATE Update $\mPhi_{k+1}$: $\mPhi_{k+1}  \leftarrow \calP_{\calS_{\kappa}} (\mPhi_k-\eta \nabla_{\mPhi}f(\mPhi_k,\mG_{k}) )$
		%\STATE {}
		\ENDWHILE
	\end{algorithmic}
\end{algorithm}
%\end{comment}

\bibliography{RefSensingMtx}
\bibliographystyle{ieeetr}

\end{document}